\documentclass[11pt, fleqn]{article}
\usepackage[english]{babel}
\usepackage[utf8]{inputenc}
\usepackage[T1]{fontenc}
\usepackage{lmodern}
\usepackage{amssymb}
\usepackage{color}
\usepackage[lmargin=1.1in,rmargin=1.1in,bottom=1.3in,top=1.3in,twoside=False]{geometry}

\usepackage{amsfonts}
\usepackage{amsmath}
\usepackage{amssymb}
\usepackage{amsthm}
\usepackage[english]{babel}
\usepackage{complexity}
\usepackage{enumerate}
\usepackage{mathrsfs}
\usepackage{graphics}
\usepackage[pagewise]{lineno}
\usepackage{mathtools}
\usepackage{microtype}
\usepackage[defaultlines=4,all]{nowidow}
\usepackage[many]{tcolorbox}
\usepackage{todonotes}
\usepackage{xcolor}
\usepackage{xspace}

\usepackage[hidelinks]{hyperref}
\usepackage[nameinlink]{cleveref}

\usetikzlibrary{calc,arrows, automata, fit, shapes, arrows.meta, positioning}

\tcolorboxenvironment{theorem}{
    colframe=cyan, sharp corners=west, 
    colback=cyan!10,
	rightrule=0pt,
	toprule=0pt,
	bottomrule=0pt,
	top=0pt,
	right=0pt,
	bottom=0pt,
 }
\tcolorboxenvironment{corollary}{
    colframe=blue, sharp corners=west, 
    colback=blue!10,
	rightrule=0pt,
	toprule=0pt,
	bottomrule=0pt,
	top=0pt,
	right=0pt,
	bottom=0pt,
 }
\tcolorboxenvironment{lemma}{
    colframe=red, sharp corners=west,
	colback=red!10,
	rightrule=0pt,
	toprule=0pt,
	bottomrule=0pt,
	top=0pt,
	right=0pt,
	bottom=0pt,
 }

 \tcolorboxenvironment{result}{
    colframe=yellow, sharp corners=west, 
    colback=yellow!10,
	rightrule=0pt,
	toprule=0pt,
	bottomrule=0pt,
	top=0pt,
	right=0pt,
	bottom=0pt,
 }

\renewcommand{\epsilon}{\varepsilon}
\renewcommand{\phi}{\varphi}

\newcommand{\Cc}{\mathscr{C}}

\newcommand{\Oof}{\mathcal{O}}
\renewcommand{\phi}{\varphi}
\renewcommand{\vec}{\mathbf}
\renewcommand{\epsilon}{\varepsilon}
\newcommand{\budget}{b}

\newcommand{\cred}{\textsf{red}}
\newcommand{\cgreen}{\textsf{green}}
\newcommand{\cblue}{\textsf{blue}}

\newcommand{\sol}{\texttt{Sol}}
\newcommand{\soli}{\sol_i}
\newcommand{\solj}{\sol_j}
\newcommand{\solh}{\sol_h}
\newcommand{\solis}{\soli[s]}
\renewcommand{\S}{\texttt{S}}
\newcommand{\Si}{\S_i}
\newcommand{\valpha}{\texttt{A}}
\newcommand{\vbeta}{\texttt{B}}

\theoremstyle{remark}
\newtheorem{theorem}{Theorem}[section]
\newtheorem{corollary}{Corollary}[section]
\newtheorem{definition}{Definition}[section]

\newtheorem{proposition}{Proposition}[section]

\crefname{corollary}{Corollary}{Corollaries}
\crefname{lemma}{Lemma}{Lemmas}
\crefname{section}{Section}{Sections}

\newtheorem{result}{}
\newtheorem*{result*}{}

\begin{document}

\title{On Solution Discovery via Reconfiguration}

\author{Michael R.~Fellows \\
University of Bergen, Norway\\
\texttt{michael.fellows@uib.no}
\and 
Mario Grobler\\ University of Bremen, Germany\\ \texttt{grobler@uni-bremen.de} 
\and 
Nicole Megow\\ University of Bremen, Germany \\\texttt{nicole.megow@uni-bremen.de}
\and 
Amer E.~Mouawad\\
American University of Beirut, Lebanon\\
\texttt{aa368@aub.edu.lb}
\and
Vijayaragunathan Ramamoorthi\\University of Bremen, Germany\\\texttt{vira@uni-bremen.de}
\and
Frances A.~Rosamond\\
University of Bergen, Norway\\
\texttt{Frances.Rosamond@uib.no}
\and
Daniel Schmand \\ University of Bremen, Germany \\\texttt{schmand@uni-bremen.de}
\and 
Sebastian Siebertz\\
University of Bremen, Germany\\
\texttt{siebertz@uni-bremen.de}}

  

\maketitle
\begin{abstract}
%
%
%
The dynamics of real-world applications and systems require efficient methods for improving infeasible solutions or restoring corrupted ones by making modifications to the current state of a system in a restricted way.
We propose a new framework of \emph{\mbox{solution} discovery via reconfiguration} for constructing a feasible solution for a given problem by executing a sequence of small modifications starting from a given state. 
Our framework integrates and formalizes different aspects of classical local search, reoptimization, 
and combinatorial reconfiguration.
We exemplify our framework on a multitude of fundamental combinatorial problems, namely \textsc{Vertex Cover}, \textsc{Independent Set}, \textsc{Dominating Set}, and \textsc{Coloring}. We study the classical as well as the parameterized complexity of the solution discovery variants of 
those problems and explore the boundary between tractable and intractable instances.

\end{abstract}

\section{Introduction}

In many dynamic real-world applications of decision-making problems, feasible solutions must be found starting from a certain predetermined system state. 
This is in contrast to the classical 
academic problems where we 
are allowed to compute a feasible solution from scratch. However, when constructing a new solution from scratch, we have no control over the difference between the current state and the new target one. 
In this work, we develop a new framework, where we aim for modifying some, possibly infeasible, solution to a feasible one via a bounded number of ``small'' modification steps.


Consider for example 
a frequency assignment problem~\cite{AardalHKMS07,Hale1980} where 
agents 
communicate via wireless message transmissions. 
Each agent can broadcast over a fixed frequency. 
Any two nearby 
agents are required to operate over different frequencies, as otherwise their signals would interfere, which must be avoided.
The objective is to assign a given number of frequencies such that no interference occurs. 
This problem can conveniently be modeled as a 
vertex coloring problem on graphs. 
The vertices of the graph represent the agents and two vertices are connected by an edge if they are  close to one another and frequencies of the corresponding agents could interfere. 
Then, the frequencies correspond to colors assigned to the vertices, as adjacent vertices must receive distinct colors.
In the solution discovery variant of the graph coloring problem, 
we are given a non-proper coloring of a graph using a fixed number of colors and a fixed budget. The goal is to 
transform the given non-proper coloring into a proper one such that the number of recoloring steps does not exceed the budget (various definitions of a recoloring step exist; we discuss the most common ones 
later). 
The above agent communication example corresponds to the setting where agents have already been assigned frequencies that cause some interference. This might happen due to various possible 
changes to a running system such as the introduction of new agents, the change of locations, or the increase of the broadcasting range. 
Recomputing a valid coloring from scratch might be undesirable as this may induce many changes in the system. 
Depending on the given infeasible assignment, it might be possible to satisfy the feasibility constraints by just a few controlled changes to the coloring. A change however, may trigger other changes, and the question is whether a ``cheap'' reconfiguration of the system into a valid state is possible. 

As a second motivation for our framework consider, e.g., a group of mobile agents tasked with monitoring (parts of) a country (represented by a graph) for security threats, such as natural disasters. 
Each agent is responsible for monitoring a certain subset of the country (a few cities), and the cities are connected in a way that reflects the dependencies/adjacencies between the different parts of the network. 
The agents in a dominating set would represent the agents that are able to monitor the entirety of the parts of interest.
By identifying a small dominating set in this graph, it is possible to identify/place a small number of agents that are responsible  for maintaining the security of the country. 
In the real world, agents can be retired from or added to the system or the areas of interest to monitor (i.e.\ cities) can also change over time. 
Given that it is most likely expensive to move the agents around (think of agents as being mobile monitoring vehicles), it is desirable to be able to compute a new dominating set in the modified system while accounting for the cost of transformation steps.

Our new framework 
models such solution discovery versions of static (decision) problems. 
Solution discovery problems are similar in flavor to  other approaches transforming one solution to another such as 
\emph{local search}, \emph{reoptimization}, \emph{dynamic algorithms}, and \emph{combinatorial reconfiguration}. 

\emph{Local search} is an algorithmic paradigm that is based on the iterative improvement of (possibly infeasible) solutions by searching for a better solution in a well-defined neighborhood (typically defined by certain local changes). A local search procedure terminates once there is no improved solution in the neighborhood. Strongly related is the framework of {\em reoptimization} where one looks for 
solutions for a problem instance when given solutions to neighboring instances. 
Local search and reoptimization have proven to be very powerful for many problems in theory and practice; see,  e.g.,~\cite{aartsL1997survey,galinier2006survey,DBLP:books/tf/18/BockenhauerHK18}.
In contrast to those models, in a solution discovery problem, 
a sequence of reconfiguration steps may not necessarily find improved solutions at each step, as long as we arrive at a feasible solution within a bounded number of steps.  
%
%

In the area of \emph{dynamic (graph) algorithms}, the dynamic nature of real-world systems is modeled  via a sequence of additions/deletions of vertices/edges in a graph. The goal 
is to output an updated solution much more efficiently than a static algorithm that is computing one from scratch. Dynamic graph algorithms have been studied extensively, also for coloring graphs~\cite{BhattacharyaCHN18}, 
maximum independent sets~\cite{
AssadiOSS18}, 
and minimum vertex covers~\cite{OnakR10,BhattacharyaHI18}. The 
(parameterized) complexity of dynamic algorithms has been studied, e.g., in~\cite{hartung2013incremental,abu2015parameterized,krithikaST2017}.
In contrast to the solution discovery framework, changes to the system happen in a very specific way, while the reconfiguration steps are restricted only indirectly by the computation time (and space) and not by a concrete reconfiguration budget. 


In the \emph{combinatorial reconfiguration} framework~\cite{van2013complexity,nishimura2018introduction}, we are given both an initial solution and a target solution (for a fixed graph) and we aim to transform the initial solution to the target one by executing a sequence of reconfiguration steps chosen from a well-defined restricted collection of steps that is typically described by allowed token moves. 
However, in a solution discovery problem, the target feasible solution is not known (and might not even exist) and the goal is to find any feasible target solution via arbitrary intermediate (possibly infeasible) configurations whose number is bounded by the budget. 

%

\subsection{Solution discovery via reconfiguration}

We focus on fundamental graph problems, where we are given a graph $G$ and an integer $k$ and the goal is to 
find a feasible solution for an instance~$(G,k)$, e.g., finding a vertex cover, independent set, dominating set of size (at most)~$k$, or a proper vertex coloring using at most $k$ colors. 
These are prominent combinatorial problems in artificial intelligence, see e.g.~\cite{CaiHLL18,HebrardK19,HuangLWZM019,LuoHCLZZ19,Shen00EE22,WangCPLY20}, with plenty of applications to real-world problems, 
including feature selection~\cite{XieQ18}, scheduling, planning, resource allocation~\cite{BlumR08}, frequency assignment~\cite{AardalHKMS07}, 
network security~\cite{CaiLL17}, sensor systems~\cite{KavalciUD14}, and many more.



We introduce the solution discovery variant of such problems, where we are given a graph~$G$, a starting configuration of size $k$ (which is not necessarily a feasible solution), and a budget~$\budget \in \mathbb{N}$. The goal is to transform the starting configuration into a feasible solution using at most $\budget$ modification steps. Inspired by the reconfiguration framework, we mainly focus on modification steps consisting of changes along edges of the graph ({\em token sliding}). We briefly discuss other modification steps ({\em token addition/removal resp.~swapping} and {\em jumping}) that have been studied in the literature. 

As mentioned earlier, our framework is inspired by methods for finding a feasible solution starting from a predefined system state such as in local search, dynamic algorithms, and reoptimization, but it restricts the local modification steps as is common in combinatorial reconfiguration. Two applications with natural restrictions on how a solution can be modified by moving an entity along edges in a graph were presented in the previous section. More generally, we model applications where items/agents can be moved only along links in a given physical network such as road network or computer network, or whenever there are restrictions on the movement trajectories. 

We demonstrate our framework by applying it to the \textsc{Vertex Cover}, \textsc{Independent Set}, \textsc{Dominating Set}, and \textsc{Coloring} problems. These are central problems that have been extensively studied from the perspective of local search, reoptimization, and combinatorial reconfiguration. To state our results, 
we assume some familiarity with graph theory, the aforementioned problems, and (parameterized) complexity. 
All notation will be formally defined before it is first used in the following sections. 

Let $\Pi$ be a vertex subset problem, i.e., a problem defined on undirected graphs such that a solution consists of a subset of vertices. The \textsc{$\Pi$-Discovery} problem is defined as follows. 
We are given a graph $G$, a subset \mbox{$S \subseteq V(G)$} of size $k$ (which at this point is not necessarily a solution for $\Pi$), and a budget $\budget$ (as a positive integer). 
We assume that each vertex in $S$ contains a token (which corresponds to an agent).
The goal is to decide if we can move the tokens on~$S$ using at most $\budget$ moves such that the resulting set is a solution for~$\Pi$.

Dependent on the underlying problem, different notions of token moves have been established in the reconfiguration literature. That is, for token configurations $Q, Q' \subseteq V(G)$ we have the following models. 
In the \emph{token sliding model}, we say a token \emph{slides} from $u \in V(G)$ to $v \in V(G)$ if $u \in Q$, $v \not\in Q$, $v \in Q'$, $u \not\in Q'$, and $\{u,v\} \in E(G)$.
In the \emph{token jumping model}, a token \emph{jumps} from $u \in V(G)$ to $v \in V(G)$ if $u \in Q$, $v \not\in Q$, $v \in Q'$, and $u \not\in Q'$.
Finally, in the \emph{token addition/removal model}, a token is \emph{added} to vertex $v \in V(G)$ if $v \not\in Q$ and $v \in Q'$. Similarly, a token is \emph{removed} from vertex $v \in V(G)$ if~$v \in Q$ and~$v \not\in Q'$.

We note that we focus on sliding for vertex subset problems as most of our results for \mbox{\textsc{$\Pi$-Discovery}} are based on it. 
%
By instantiating~$\Pi$ with 
\textsc{Vertex Cover}, \textsc{Independent Set}, or \textsc{Dominating Set}, we obtain the \textsc{Vertex Cover Discovery} (VCD), \textsc{Independent Set Discovery} (ISD), or \textsc{Dominating Set Discovery} (DSD) problem, respectively.
%
Similar to the vertex subset discovery problems, we define the \textsc{Coloring Discovery} problem in~\cref{sec-coloring}.

\subsection{Our results}

We study the classical as well as the parameterized complexity of the aforementioned solution discovery problems and prove the following results. 

\begin{result}
\textbf{\textsc{Vertex Cover Discovery}:} We show that the problem is \textsf{NP}-complete even on planar graphs of maximum degree four and \textsf{W[1]}-hard for parameter $\budget$, even on $2$-degenerate bipartite graphs. On the positive side, we show that the problem is \textsf{XP} for parameter treewidth (that is, polynomial-time solvable on every fixed class of bounded treewidth), \textsf{FPT} for parameter $k$ on general graphs and \textsf{FPT} for parameter~$\budget$ when restricted to structurally nowhere dense classes of graphs.
\end{result}

\begin{result}
\textbf{\textsc{Independent Set Discovery}:}
We show that the problem is \textsf{NP}-complete even on planar graphs of maximum degree four, \textsf{W[1]}-hard for parameter $k + \budget$ even on graphs excluding $\{C_4, \ldots, C_p\}$ as induced subgraphs (for any constant $p$), and  \textsf{W[1]}-hard for parameter $\budget$ even on $2$-degenerate bipartite graphs. On the positive side, we show that the problem is \textsf{XP} for parameter treewidth, \textsf{FPT} for parameter $k$ on $d$-degenerate and nowhere dense classes of graphs as well as \textsf{FPT} for parameter $\budget$ on structurally nowhere dense classes of graphs. 
\end{result}

\begin{result}\textbf{\textsc{Dominating Set Discovery}:}
We show that the problem is \textsf{NP}-complete even on planar graphs of maximum degree five, \textsf{W[1]}-hard for parameter $k + \budget$ even on bipartite graphs, and \textsf{W[1]}-hard for parameter $\budget$ even on $2$-degenerate graphs. On the positive side, we show that the problem is \textsf{XP} for parameter treewidth, \textsf{FPT} for parameter~$k$ on biclique-free and semi-ladder-free graphs as well as \textsf{FPT} for parameter~$\budget$ on structurally nowhere dense classes of graphs.  
\end{result}

\begin{result}
\textbf{\textsc{Coloring Discovery}:}
We show that the problem is \textsf{NP}-complete for every $k\geq3$ even on planar bipartite graphs (that is, para-\textsf{NP}-hard for parameter $k$), \textsf{W[1]}-hard for parameter treewidth, and \textsf{W[1]}-hard for parameter $k+\budget$ on general graphs. On the positive side, we show that the problem is polynomial-time solvable for $k=2$ and \textsf{FPT} parameterized by $k+\budget$ on structurally nowhere dense classes of graphs.
\end{result}

Our results provide an almost complete classification of the complexity of the problems on minor-closed, resp.\ monotone (i.e.\ subgraph-closed) classes of graphs. Concerning the polynomial-time solvable cases, observe that a class of graphs has bounded treewidth if and only if it excludes a planar graph as a minor~\cite{DBLP:journals/jct/RobertsonS86}. Hence, for the vertex subset discovery problems, hardness on planar graphs implies that for minor-closed classes of graphs efficient algorithms can only be expected for classes with bounded treewidth, which is exactly what we establish. \textsc{Coloring Discovery} is \textsf{NP}-complete both on planar graphs and on classes with bounded treewidth for any $k\geq 3$, hence, our tractability result for $k=2$ is the best we can hope for. 

\begin{figure}
\centering
    \begin{tikzpicture}[->, scale=.75, auto=left, remember picture,every node/.style={rectangle},inner/.style={rectangle},outer/.style={rectangle}]


    \node[outer] (TREEWIDTH) at (9,0) [rectangle, fill=white, draw=black] {Bounded treewidth};
    \node[outer] (PLANAR) at (4,0) [rectangle, fill=white, draw=black] {Planar};
    \node[outer] (MINOR) at (6,2) [rectangle, fill=white, draw=black] {$H$-minor-free};
    \node[outer] (DENSE) at (9,4) [rectangle, fill=white, draw=black] {Nowhere dense};
    \node[outer] (DEGENERATE) at (3,4) [rectangle, fill=white, draw=black] {Bounded degeneracy};
    \node[outer] (BICLIQUE) at (6,6) [rectangle, fill=white, draw=black] {Biclique-free};
    \node[outer] (SEMI) at (6,8) [rectangle, fill=white, draw=black] {Semi-ladder-free};
    \node[outer] (SDENSE) at (12,6) [rectangle, fill=white, draw=black] {Structurally nowhere dense};

    \foreach \from/\to in {PLANAR.north/MINOR,TREEWIDTH/MINOR,MINOR/DEGENERATE,MINOR/DENSE, DEGENERATE/BICLIQUE, DENSE/BICLIQUE, BICLIQUE/SEMI, DENSE/SDENSE}
    \draw[-stealth] (\from) -> (\to);

    \end{tikzpicture}
\caption{The considered graph classes. Arrows indicate inclusion.}
\label{fig-graph-classes}
\end{figure}
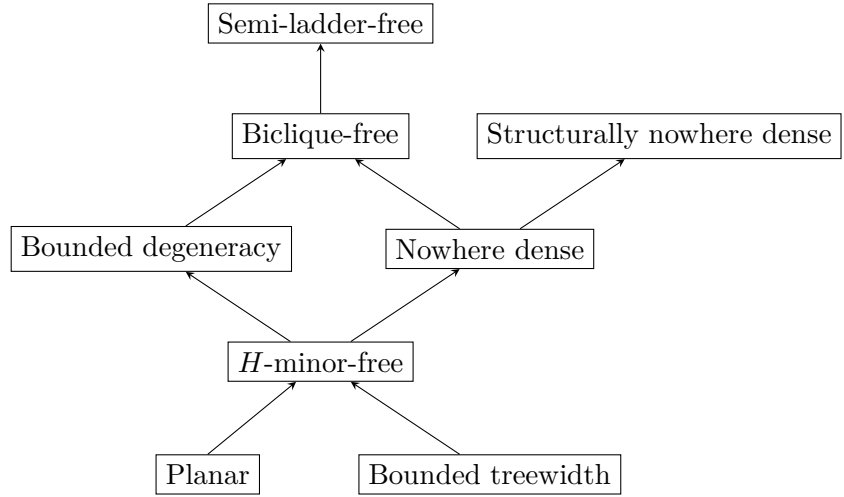

The notions of degeneracy and nowhere density are very general notions of graph sparsity. Besides the theoretical interest in these classes, many real-world networks, such as several social networks (such as physical disease propagation networks), biological networks
(such as gene interactions or brain networks), and informatics networks (such as autonomous systems), turn out to be degenerate or nowhere dense, see e.g.~\cite{brown2020exploring, demaine2019structural,farrell2015hyperbolicity}.
Degenerate and nowhere dense classes are biclique free. It is not easy to find biclique-free classes that are neither degenerate nor nowhere dense -- one interesting example being the class of all $n$-vertex graphs of girth at least $5$ with average degree~$\Oof(n^{1/5})$, see~\cite{simonovits1983extremal}. 
Structurally nowhere dense classes are a dense extension of nowhere dense classes and semi-ladder free classes are a dense extension of biclique-free classes. We refer to~\Cref{fig-graph-classes} for an overview.

We provide an almost complete classification on monotone classes of graphs for the vertex subset discovery problems and the parameter $\budget$; we establish hardness on degenerate classes and tractability on (structurally) nowhere dense classes. 
For parameter~$k$, \textsc{Vertex Cover Discovery} is tractable on all graphs. Similarly, for parameter~$k$, \textsc{Independent Set Discovery} is tractable both on degenerate and on (structurally) nowhere dense classes, and it remains open whether it is tractable on biclique-free classes. \textsc{Dominating Set Discovery} is tractable even on semi-ladder-free, and hence on biclique-free classes. 

For \textsc{Coloring Discovery} more questions remain open. It remains open whether the problem is tractable for parameter $k+\budget$ on $d$-degenerate graphs, parameter $\budget$ on planar graphs, and parameter $k$ on classes with bounded treewidth. 


\section{Preliminaries}\label{sec:prelims}

We denote the set of natural numbers by $\mathbb{N}$.
For $k \in \mathbb{N}$ we define $[k] = \{1, 2, \dots, k\}$.

\paragraph*{Graphs}
We consider finite, simple, loopless, and undirected graphs. For a graph $G$, we denote by $V(G)$ and $E(G)$ the vertex set and edge set of $G$, respectively. Two vertices $u,v\in V(G)$ with $\{u,v\}\in E(G)$ are called adjacent or neighbors. 
The vertices~$u$ and~$v$ are called the endpoints of the edge $\{u,v\}$. 
The degree of a vertex $v$ is equal to the number of  neighbors of $v$. A graph has maximum degree $d$ if all its vertices have degree at most $d$. 

A sequence $v_1,\ldots, v_q$ of pairwise distinct vertices is a path of length $q-1$ if $\{v_i,v_{i+1}\}\in E(G)$ for all $1\leq i< q$. 
A sequence $v_1,\ldots, v_q$ of pairwise distinct vertices is a cycle of length $q$ if $\{v_i,v_{i+1}\}\in E(G)$, for all $1\leq i< q$,  and $\{v_q, v_1\}\in E(G)$. We write $P_q$ to denote a path of length~$q$ and $C_q$ to denote a cycle of length $q$. 



\paragraph*{Parameterized complexity}
A \emph{parameterized problem} is a language $L\subseteq \Sigma^*\times \mathbb{N}$, where $\Sigma$ is a fixed finite alphabet. For an instance $(x,\kappa)\in \Sigma^*\times \mathbb{N}$, $\kappa$ is called the \emph{parameter}.
The problem $L$ is called \emph{fixed-parameter tractable}, \textsf{FPT} for short, if there exists an algorithm that on input $(x,\kappa)$ decides in time $f(\kappa) \cdot |(x,\kappa)|^c$ whether $(x,\kappa)\in L$, for a computable function $f$ and constant~$c$. Likewise, the problem belongs to \textsf{para-NP} if it can be solved within the same time bound by a nondeterministic algorithm. $L$ is \emph{slice-wise polynomial}, \textsf{XP} for short, if there is an algorithm deciding whether $(x,\kappa)$ belongs to $L$ in time $f(\kappa)\cdot |(x,\kappa)|^{g(\kappa)}$, for computable functions $f,g$.

The \emph{\textsf{W}-hierarchy} is a collection of parameterized complexity classes $\textsf{FPT} \subseteq \textsf{W[1]} \subseteq \ldots \subseteq \textsf{W[t]}\subseteq \ldots\subseteq \textsf{para-NP}\cap \textsf{XP}$. We have $\textsf{FPT}\neq\textsf{para-NP}$ if and only if $\textsf{P}\neq\textsf{NP}$, which is a standard assumption. Also the inclusion $\textsf{FPT}\subseteq \textsf{W[1]}$ is conjectured to be strict (and this conjecture is known to be true when assuming the exponential-time hypothesis). 
Therefore, showing intractability in the parameterized setting is usually accomplished by establishing an \textsf{FPT}-reduction from a \textsf{W}-hard problem.  

Let $L,L'\subseteq \Sigma^*\times\mathbb{N}$ be parameterized problems. A \emph{parameterized reduction} from $L$ to $L'$ is an algorithm that, given an instance $(x, \kappa)$ of $L$, outputs an instance $(x',\kappa')$ of $L'$ such that 
$(x, \kappa)\in L \Leftrightarrow (x',\kappa')\in L'$, $\kappa'\leq g(\kappa)$ for some computable function $g$, and the running time of the algorithm is bounded by $f(\kappa) \cdot |(x,\kappa)|^c$ for some computable function $f$ and constant $c$.

We refer to the textbooks~\cite{cygan2015parameterized,DBLP:series/mcs/DowneyF99,FlumGrohe2006} for extensive background on parameterized complexity.

\section{Vertex cover discovery}

In the \textsf{NP}-complete \textsc{Vertex Cover (VC)} problem~\cite{GareyJ79}, we are given a graph $G$ and an integer~$k$ and the problem is to decide whether~$G$ admits a vertex cover $C$ of size at most~$k$, i.e., a set $C$ of at most~$k$ vertices such that for every edge $\{u,v\}\in E(G)$ we have $\{u,v\}\cap C\neq \emptyset$. 

\subsection{Related work}
The \textsc{Vertex Cover} problem admits a $2$-approximation and is hard to approximate within a factor of $(2-\epsilon)$, for any $\epsilon>0$, assuming the unique games conjecture~\cite{DBLP:conf/coco/KhotR03}. 
For an integer $\lambda$, we call two sets $A$ and $B$ $\lambda$-close if their symmetric difference has size at most $\lambda$. 
Local search by exchanging with $\lambda$-close solutions, for $\lambda\in \Oof(1/poly(\epsilon))$, leads to a $(1+\epsilon)$-approximation algorithm for \textsc{Vertex Cover} on graph classes with subexponential expansion~\cite{har2017approximation}. 
The \textsc{Vertex Cover} problem is the textbook example of a fixed-parameter tractable problem~\cite{DBLP:series/mcs/DowneyF99}. The dynamic variant of the problem was studied, e.g., in~\cite{abu2015parameterized}. 

The \textsc{Vertex Cover} problem is also  very well studied under the combinatorial reconfiguration framework~\cite{van2013complexity,nishimura2018introduction}. Recall that in the reconfiguration variant of a (graph vertex subset) problem, we are given a graph $G$ and two feasible solutions~\mbox{$S$ (source)} and~\mbox{$T$ (target)} and the goal is to decide whether we can transform $S$ to $T$ via ``small'' reconfiguration steps while maintaining feasibility (sometimes bounding the number of allowed reconfiguration steps).

In contrast to the decision variant, \textsc{Vertex Cover Reconfiguration (VCR)} is known to be \textsf{PSPACE}-complete on general graphs for the token sliding, token jumping, and token addition/removal models~\cite{DBLP:journals/tcs/ItoDHPSUU11,DBLP:journals/tcs/KaminskiMM12}. 
This remains true even for very restricted graph classes such as graphs of bounded bandwidth/pathwidth/treewidth~\cite{DBLP:journals/jcss/Wrochna18}. 
On the positive side, polynomial-time algorithms are known only for very simple graph classes such as trees~\cite{DBLP:conf/isaac/DemaineDFHIOOUY14}. More positive results (which vary depending on the model) are possible if we consider the parameterized complexity of the problem. We refer the reader to~\cite{bousquet2022survey} for more details. 

\subsection{Notation and definitions}

In the \textsc{Vertex Cover Discovery (VCD)} problem, we are given a graph $G$, a starting configuration $S$, and a positive integer budget $\budget$. 
The goal is to decide whether we can reach a vertex cover of $G$ (starting from $S$) using at most $\budget$ token slides (where we forbid multiple tokens occupying the same vertex). 
Formally, an instance of \textsc{VCD} is a tuple $((G,S,\budget),\kappa)$, where $\kappa$ is the parameter. 
For readability we will simply write $(G,S,\budget)$ and make the parameter explicit in the text. We will always use $k$ to denote the size of $S$, i.e., $k = |S|$, except in \cref{sec-coloring} where $k$ will denote the number of colors. 
Possible choices for the parameter $\kappa$ are $k$, $\budget$, $k+\budget$, or structural parameters such as the treewidth of the input graph.


\subsection{Intractability}
Note that for all solution discovery variants of (graph) vertex subset problems, we can always assume that $\budget \leq n^2$, where $n$ is the number of vertices in the input graph. This follows from the fact that each token will have to traverse a path of length at most~$n$. Hence, all the solution discovery variants of such problems are indeed in $\textsf{NP}$. We first show that the \textsc{Vertex Cover Discovery} problem is \textsf{NP}-complete even on very restricted graph classes. 

\begin{theorem}
\label{thm:VCD_NPcomplete}
The \textsc{Vertex Cover Discovery} problem is \textsf{NP}-complete on the class of planar graphs of maximum degree four. 
\end{theorem}
\begin{proof}
We give a reduction from \textsc{VC} on planar graphs of maximum degree three, which is known to be \textsf{NP}-complete~\cite{DBLP:journals/siamcomp/Lichtenstein82}. Given an instance $(G,\kappa)$ of \textsc{VC}, where $G$ is a planar graph of maximum degree three, we construct an instance of \textsc{VCD} as follows. 
We create a new graph $H$ which consists of a copy of $G$. Then, for each vertex $v \in V(H)$, we create a new path consisting of three vertices $\{x_v, y_v, z_v\}$ and we connect $v$ to $x_v$. We choose $S = \{x_v, y_v \mid v \in V(G)\}$ and we set the budget $\budget = \kappa$. Note that $k = |S| = 2|V(G)|$. 
This completes the construction of the instance $(H,S,\budget)$ of \textsc{VCD}. It follows from the construction that $H$ is planar and of maximum degree four. We prove that $(G,\kappa)$ is a yes-instance of \textsc{VC} if and only if $(H,S,\budget)$ is a yes-instance of \textsc{VCD}.

First, assume that $G$ has a vertex cover $C$ of size at most~$\kappa$. Then, in $H$ we can simply slide every token on $x_v$ to $v$, where $v \in C$. Since $C$ is of size at most $\kappa = \budget$, we have at most~$\budget$ slides. To see that the resulting set is a vertex cover of $H$, note that every edge $\{v, x_v\}$ is still covered by either $v$ or $x_v$ and the edges $\{x_v, y_v\}$ and $\{y_v, z_v\}$ are still covered by $y_v$. Moreover, all the other edges of $H$ are covered since $C$ is a vertex cover of~$G$. 

For the reverse direction, assume that $(H,S,\budget)$ is a yes-instance of \textsc{VCD}. Since $\budget = \kappa$, we know that at most $\kappa$ tokens can move, i.e., be placed in $H[V(G)]$. Moreover, since the resulting configuration must be a vertex cover of $H$ and $H[V(G)]$ is an induced subgraph of $H$ it follows that the tokens on vertices corresponding to vertices of $G$ must form a vertex cover of $G$ of size at most $\kappa$, as needed. 
\end{proof}

Our next result shows that, from a parameterized perspective, the \textsc{Vertex Cover Discovery} problem remains hard when parameterized by the budget $\budget$ alone, even on $2$-degenerate bipartite graphs.  
Recall that a graph~$G$ is called $d$-degenerate if every subgraph $H$ of $G$ has a vertex of degree at most $d$.

\begin{theorem}
\label{thm:VCD_WHARD}
The \textsc{Vertex Cover Discovery} problem is \textsf{W[1]}-hard when parameterized by the budget $\budget$ even on the class of $2$-degenerate bipartite graphs.
\end{theorem}

\begin{proof}
We present a parameterized reduction from the \textsc{Clique} problem. 
Let $(G, \kappa)$ be an instance of the \textsc{Clique} problem. 
We construct a graph $H$ from $G$ as follows. 
First, we define the vertex set $V(H)$. 
$V(H)$ is partitioned into six vertex sets $V_1,\ldots, V_6$. 
Let $V_1 = \{r_i \mid i \in [\binom{\kappa}{2}]\}$ and $V_2 = \{z_i \mid i \in [\binom{\kappa}{2}]\}$. For $m=|E(G)|$, the set $V_3$ consists of $2m\binom{\kappa}{2}$ vertices grouped into $m$ sets $V_3^e$ for each $e \in E(G)$. 
For each $e \in E(G)$, let $V_3^e = \{g_e^i, h_e^i \mid i \in [\binom{\kappa}{2}]\}$. 
Let $V_4 = \{y_e \mid e \in E(G)\}, V_5 = \{x_u \mid u \in V(G)\}$ and $V_6 = \{s_u \mid u \in V(G)\}$. 
Now we define the edge set $E(H)$. 
For each $i \in [\binom{\kappa}{2}]$, we add an edge between the vertices $r_i \in V_1$ and $z_i \in V_2$. 
Next, for each $i \in [\binom{\kappa}{2}]$ and $e \in E(G)$, the vertex $g_e^i \in V_3^e$ is adjacent to the vertices $z_i \in V_2$, $h_e^i \in V_3^e$, and $y_e \in V_4$. 
For each edge $e = \{u,v\} \in E(G)$, the vertex $y_e \in V_4$ is adjacent to $x_u, x_v \in V_5$. 
Finally, for each $u \in V(G)$, we add an edge between the vertices $x_u \in V_5$ and $s_u \in V_6$. 
This completes the construction of the graph $H$. We define the starting configuration $S = V_6 \cup V_4 \cup \{g_e^i \in V_3 \mid e \in E(G), i\in [\binom{\kappa}{2}]\}$ of size $k = m+n+m\binom{\kappa}{2}$ and we set $\budget = 2\binom{\kappa}{2} + \kappa$~(see~\Cref{fig:my_label}).

We first prove that the graph $H$ is indeed $2$-degenerate. Recall that a graph $H$ is $d$-degenerate if every induced subgraph of $H$ contains a vertex of degree at most $d$. Consider any induced subgraph $H'$ of $H$. If $H'$ contains a vertex from $V_1$, $V_3$, or $V_6$, then we are done, since vertices in $V_1$ and $V_6$ have degree one in $H$ and vertices of $V_3$ have degree three in $H$ but one of those neighbors is a vertex of degree one. It is not hard to see that any induced subgraph of $H$ not including a vertex of $V_1 \cup V_3 \cup V_6$ must either have an isolated vertex from $V_2$, or a degree-two vertex from $V_4$, or an isolated vertex from~$V_5$ (when no vertex from $V_4$ is included). 

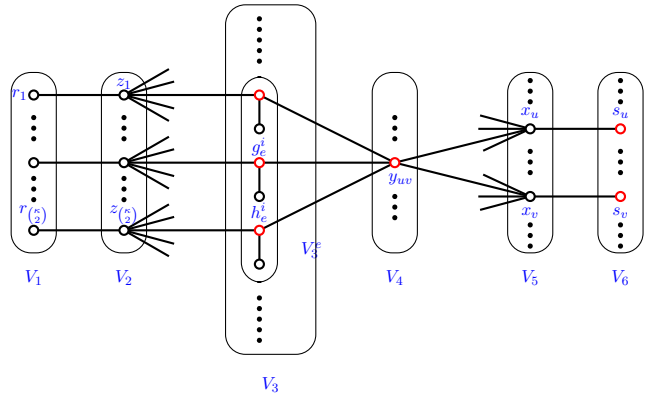
\begin{figure}[ht]
\centering
    \begin{tikzpicture}[scale=0.6, every node/.style={transform shape}]
    \makeatletter
    \tikzset{
        dot diameter/.store in=\dot@diameter,
        dot diameter=3pt,
        dot spacing/.store in=\dot@spacing,
        dot spacing=10pt,
        dots/.style={
            line width=\dot@diameter,
            line cap=round,
            dash pattern=on 0pt off \dot@spacing
        }
    }
    \makeatother

    \coordinate (o) at (0,0);
    \coordinate (r1) at ($(o) + (-5,3)$);
    \coordinate (r2) at ($(o) + (-5,1.5)$);
    \coordinate (r3) at ($(o) + (-5,0)$);
    \coordinate (z1) at ($(o) + (-3,3)$);
    \coordinate (z2) at ($(o) + (-3,1.5)$);
    \coordinate (z3) at ($(o) + (-3,0)$);
    \coordinate (a1) at ($(o) + (0,3)$);
    \coordinate (a2) at ($(o) + (0,1.5)$);
    \coordinate (a3) at ($(o) + (0,0)$);
    \coordinate (y1) at ($(o) + (3,3)$);
    \coordinate (y2) at ($(o) + (3,1.5)$);
    \coordinate (y3) at ($(o) + (3,0)$);
    \coordinate (x1) at ($(o) + (6,3)$);
    \coordinate (x2) at ($(o) + (6,2.25)$);
    \coordinate (x3) at ($(o) + (6,0.75)$);
    \coordinate (x4) at ($(o) + (6,0)$);
    \coordinate (s1) at ($(o) + (8,3)$);
    \coordinate (s2) at ($(o) + (8,2.25)$);
    \coordinate (s3) at ($(o) + (8,0.75)$);
    \coordinate (s4) at ($(o) + (8,0)$);

    \draw[black, thick] (z1) -- (r1);
    \draw[black, thick] (z2) -- (r2);
    \draw[black, thick] (z3) -- (r3);

    \draw[black, thick] (z1) -- ($(z1)+(15:1.15cm)$);
    \draw[black, thick] (z1) -- ($(z1)+(30:1.15cm)$);
    \draw[black, thick] (z1) -- ($(z1)+(-15:1.15cm)$);
    \draw[black, thick] (z1) -- ($(z1)+(-30:1.15cm)$);
    \draw[black, thick] (z2) -- ($(z2)+(15:1.15cm)$);
    \draw[black, thick] (z2) -- ($(z2)+(30:1.15cm)$);
    \draw[black, thick] (z2) -- ($(z2)+(-15:1.15cm)$);
    \draw[black, thick] (z2) -- ($(z2)+(-30:1.15cm)$);
    \draw[black, thick] (z3) -- ($(z3)+(15:1.15cm)$);
    \draw[black, thick] (z3) -- ($(z3)+(30:1.15cm)$);
    \draw[black, thick] (z3) -- ($(z3)+(-15:1.15cm)$);
    \draw[black, thick] (z3) -- ($(z3)+(-30:1.15cm)$);

    \draw[black, thick] (a1) -- (z1);
    \draw[black, thick] (a2) -- (z2);
    \draw[black, thick] (a3) -- (z3);
    \draw[black, thick] (a1) -- ($(a1) + (0,-0.75)$);
    \draw[black, thick] (a2) -- ($(a2) + (0,-0.75)$);
    \draw[black, thick] (a3) -- ($(a3) + (0,-0.75)$);

    \draw[black, thick] (y2) -- (a1);
    \draw[black, thick] (y2) -- (a2);
    \draw[black, thick] (y2) -- (a3);

    \draw[black, thick] (y2) -- (x2);
    \draw[black, thick] (y2) -- (x3);

    \draw[black, thick] (x2) -- ($(x2)+(165:1.15cm)$);
    \draw[black, thick] (x2) -- ($(x2)+(180:1.15cm)$);
    \draw[black, thick] (x2) -- ($(x2)+(205:1.15cm)$);
    \draw[black, thick] (x3) -- ($(x3)+(155:1.15cm)$);
    \draw[black, thick] (x3) -- ($(x3)+(180:1.15cm)$);
    \draw[black, thick] (x3) -- ($(x3)+(195:1.15cm)$);

    \draw[black, thick] (x2) -- (s2);
    \draw[black, thick] (x3) -- (s3);

    \draw[black, thick, fill=white] (r1) circle (0.1cm) node[left, blue] {$r_1$};
    \draw[black, thick, fill=white] (r2) circle (0.1cm); 
    \draw[black, thick, fill=white] (r3) circle (0.1cm) node[above, blue] {$r_{\binom{\kappa}{2}}$};
    \draw[black, thick, fill=white] (z1) circle (0.1cm) node[above, blue] {$z_1$};
    \draw[black, thick, fill=white] (z2) circle (0.1cm);
    \draw[black, thick, fill=white] (z3) circle (0.1cm) node[above, blue] {$z_{\binom{\kappa}{2}}$};
    \draw[red, thick, fill=white] (a1) circle (0.1cm);
    \draw[red, thick, fill=white] (a2) circle (0.1cm) node[above, blue] {$g_e^i$};
    \draw[red, thick, fill=white] (a3) circle (0.1cm);
    \draw[black, thick, fill=white] ($(a1) + (0,-0.75)$) circle (0.1cm);
    \draw[black, thick, fill=white] ($(a2) + (0,-0.75)$) circle (0.1cm) node[below, blue] {$h_e^i$};
    \draw[black, thick, fill=white] ($(a3) + (0,-0.75)$) circle (0.1cm);
    \draw[red, thick, fill=white] (y2) circle (0.1cm) node[below=3pt, blue] {$~~y_{uv}$};
    \draw[black, thick, fill=white] (x2) circle (0.1cm) node[above=3pt, blue] {$x_u$};
    \draw[black, thick, fill=white] (x3) circle (0.1cm) node[below=3pt, blue] {$x_v$};
    \draw[red, thick, fill=white] (s2) circle (0.1cm) node[above=3pt, blue] {$s_u$};
    \draw[red, thick, fill=white] (s3) circle (0.1cm) node[below=3pt, blue] {$s_v$};

    \draw [black, dot diameter=2pt, dot spacing=4pt, dots] ($(r1) + (0, -0.5)$) -- ($(r2) + (0, 0.5)$);
    \draw [black, dot diameter=2pt, dot spacing=4pt, dots] ($(r2) + (0, -0.35)$) -- ($(r3) + (0, 0.65)$);
    \draw [black, dot diameter=2pt, dot spacing=4pt, dots] ($(z1) + (0, -0.5)$) -- ($(z2) + (0, 0.5)$);
    \draw [black, dot diameter=2pt, dot spacing=4pt, dots] ($(z2) + (0, -0.35)$) -- ($(z3) + (0, 0.65)$);
    \draw [black, dot diameter=2pt, dot spacing=4pt, dots] ($(y1) + (0, -0.5)$) -- ($(y2) + (0, 0.5)$);
    \draw [black, dot diameter=2pt, dot spacing=4pt, dots] ($(y2) + (0, -0.75)$) -- ($(y3) + (0, 0.25)$);
    \draw [black, dot diameter=2pt, dot spacing=4pt, dots] ($(a1) + (0, 0.75)$) -- ($(a1) + (0, 1.75)$);
    \draw [black, dot diameter=2pt, dot spacing=4pt, dots] ($(a3) + (0, -1.5)$) -- ($(a3) + (0, -2.5)$);
    \draw [black, dot diameter=2pt, dot spacing=4pt, dots] ($(x2) + (0, -0.5)$) -- ($(x3) + (0, 0.5)$);
    \draw [black, dot diameter=2pt, dot spacing=4pt, dots] ($(x2) + (0, 0.6)$) -- ($(x2) + (0, 1.1)$);
    \draw [black, dot diameter=2pt, dot spacing=4pt, dots] ($(x3) + (0, -1.1)$) -- ($(x3) + (0, -0.6)$);
    \draw [black, dot diameter=2pt, dot spacing=4pt, dots] ($(s2) + (0, -0.5)$) -- ($(s3) + (0, 0.5)$);
    \draw [black, dot diameter=2pt, dot spacing=4pt, dots] ($(s2) + (0, 0.6)$) -- ($(s2) + (0, 1.1)$);
    \draw [black, dot diameter=2pt, dot spacing=4pt, dots] ($(s3) + (0, -1.1)$) -- ($(s3) + (0, -0.6)$);

    \draw[black,rounded corners=7.5pt] ($(s1)+(-0.5,0.5)$) rectangle ($(s4)+(0.5,-0.5)$) node[below=2.25cm, midway, blue] {$V_6$};
    \draw[black,rounded corners=7.5pt] ($(x1)+(-0.5,0.5)$) rectangle ($(x4)+(0.5,-0.5)$) node[below=2.25cm, midway, blue] {$V_5$};
    \draw[black,rounded corners=7.5pt] ($(y1)+(-0.5,0.5)$) rectangle ($(y3)+(0.5,-0.5)$) node[below=2.25cm, midway, blue] {$V_4$};
    \draw[black,rounded corners=7.5pt] ($(a1)+(-0.4,0.4)$) rectangle ($(a3)+(0.4,-1.15)$) node[below=1.25cm, midway, blue] {$~~~~~~~~~~~~~~~~~~V_3^e$};
    \draw[black,rounded corners=7.5pt] ($(a1)+(-0.75,2)$) rectangle ($(a3)+(1.25,-2.75)$) node[below=4.25cm, midway, blue] {$V_3$};
    \draw[black,rounded corners=7.5pt] ($(z1)+(-0.5,0.5)$) rectangle ($(z3)+(0.5,-0.5)$) node[below=2.25cm, midway, blue] {$V_2$};
    \draw[black,rounded corners=7.5pt] ($(r1)+(-0.5,0.5)$) rectangle ($(r3)+(0.5,-0.5)$) node[below=2.25cm, midway, blue] {$V_1$};
    
\end{tikzpicture}
\caption{An illustration of the \textsf{W[1]}-hardness reduction. The red-colored vertices denote the starting configuration.}
\label{fig:my_label}
\end{figure}

We now prove the correctness of the reduction starting with the forward direction. 
Let $C$ be a clique of size $\kappa$ in~$G$. 
Let $S' = (S \cup V_2 \cup \{x_u \mid u \in C\}) \setminus (\{y_{uv} \mid u,v \in C\} \cup \{s_u \mid u \in C\})$. 
We obtain $S' \subseteq V(H)$ from $S$ using at most~$\budget$ token slides. 
We slide the tokens from  any $\binom{\kappa}{2}$ arbitrary vertices in $V_3$ to $V_2$.
Then, we will slide the tokens from $\{y_{uv} \mid u,v \in C\} \subseteq V_4$ to the swapped $\binom{\kappa}{2}$ vertices in $V_3$. 
Next, we slide $\kappa$ tokens from $\{s_u \mid u \in C\}$ to $\{x_u \mid u \in C\}$. 
Since $V_2 \cup \{g_e^i \in V_3 \mid e \in E(G), i\in [\binom{\kappa}{2}]\} \subseteq S'$, the edges incident on these vertices are covered. 
For each $u \in V(G)$, either $x_u \in V_5$ or $s_u \in V_6$ is in $S'$. Therefore, the edges between the sets $V_5$ and~$V_6$ are covered by $S'$. 
Finally, the edges between the vertices in~$V_4$ that are not in $S'$ and~$V_5$ are covered by the vertices corresponding to the clique $C$ since the moved vertices in $V_4$ correspond to the edges of the clique $C$. 

For the reverse direction, we consider the properties of feasible solutions for the \textsc{VCD} problem in~$H$. 
The set $S$ does not contain any vertex from  $V_1 \cup V_2$. Since the edges between $V_1$ and $V_2$ are $\binom{\kappa}{2}$ matching edges, we need to move $\binom{\kappa}{2}$ tokens from $S$ to $V_2$. 
These tokens have to come from $V_3$. 
For each $i \in [\binom{\kappa}{2}]$ and each $e \in E(G)$, either $g_e^i$ or $h_e^i$ should be in a feasible solution. 
Therefore, we have to move $\binom{\kappa}{2}$ tokens from $V_4$ to $V_3$. We are left with only a budget of $\kappa$ to cover the edges between the $\binom{\kappa}{2}$ vertices. 
This is possible only when the $\binom{\kappa}{2}$ edges form a clique. 
\end{proof}

\subsection{Tractability}

We now show that the \textsc{Vertex Cover Discovery} problem becomes fixed-parameter tractable when considering different parameters. Unlike~\Cref{thm:VCD_WHARD}, our next result shows that parameterizing by the number of tokens $k$ instead of the budget $b$ makes the problem tractable. 

\begin{theorem}\label{thm-vc-fpt-k}
The \textsc{Vertex Cover Discovery} problem is fixed-parameter tractable when parameterized by the number of tokens $k$. Moreover, the problem can be solved in time $\mathcal{O}(2^k \cdot n^3)$.
\end{theorem}

\begin{proof}
Our algorithm consists of an enumeration step, which enumerates all minimal vertex covers of $G$ of size at most~$k$. Then, for each vertex cover $C$ of size at most $k$, we solve an instance of \textsc{Minimum Weight Perfect Matching (MWPM)} on bipartite graphs. 

To enumerate all vertex covers of size at most $k$, we follow the standard search-tree algorithm that gradually constructs
minimal vertex covers,
producing all minimal vertex covers of size at most $k$ in its leaves. Consider an instance of \textsc{VC}. To build the search-tree, at each non-leaf
node, the algorithm chooses an edge that has not yet been covered, and branches on the two possible
 ways of covering this edge, including one of the endpoints of the edge in each branch. 
 Since we are not interested in
vertex covers of cardinality more than~$k$, we do not need to search
beyond depth $k$ in the tree, proving an upper bound of~$2^k$ on the number
of leaves, and an upper bound of $\mathcal{O}^*(2^k)$ on the enumeration time. 

It remains to check whether any vertex cover $C$ of size at most $k$ can be reached from~$S$ using at most~$\budget$ token slides. To do so, we proceed as follows. For each set $C$, we construct a complete weighted bipartite graph $H$ where the bipartition of $H$ consists of $(S,C)$. 
That is, we have the vertices of~$S$ on one side and the vertices of $C$ on the other side. 
For each pair of vertices $u \in S$ and $v \in C$, we set the weight of the edge $\{u,v\} \in E(H)$ to the number of edges along a shortest path from $u$ to $v$ in $G$ (we set the weight to $m + \budget + 1$ whenever~$u$ and $v$ belong to different connected components). 
Now we apply an algorithm for \textsc{MWPM} in $H$. If the weight of the perfect matching is $\budget$ or less, then we can move the tokens accordingly. Otherwise, we proceed to the next vertex cover. To see why this is indeed correct, note that if the path a token~$t_1$ would
take to reach its destination has another token $t_2$ on it, we can have the two tokens switch destinations and continue
moving $t_2$. One can check that the number of moves does not exceed the weight of the perfect matching. 
On the other hand, any solution to the discovery problem must move tokens to targets and cannot do better than the total length of the shortest paths in a minimum weight perfect matching. This observation was first made in~\cite{DBLP:journals/siamdm/CalinescuDP08}. It is well-known that the \textsc{MWPM} problem can be solved in $\mathcal{O}(n^3)$ time using, e.g., Edmonds' blossom algorithm~\cite{edmonds1972theoretical}.
\end{proof}

Next, we show that \textsc{VCD} is \textsf{XP} when parameterized by the treewidth of the input graph. In particular, the problem is polynomial-time solvable on every fixed class of bounded treewidth. Note that $k$ in the following theorem can be in $\Oof(n)$ so that $2^{\mathcal{O}(t \log k)}\in n^{\Oof(t)}$. 
We remark that the problem could potentially be even \textsf{FPT} for parameter treewidth, but we were neither able to prove this claim nor prove $\textsf{W}[1]$-hardness. 

\begin{theorem}\label{thm-vc-fpt-by-tw+k}
The \textsc{Vertex Cover Discovery} problem can be solved in time $2^{\mathcal{O}(t \log k)} \cdot n^{\mathcal{O}(1)}$, where $t$ denotes the treewidth of the input graph. 
\end{theorem}

\begin{proof}
Let $(G, S, \budget)$ be an instance of the \textsc{VCD} problem. 
We first compute a nice tree decomposition $(T,(X_i)_{i\in V(T)})$ of $G$ of width $\Oof(t)$. This is possible in time $2^{\Oof(t)}\cdot n^2$ and we refer to the textbook~\cite{cygan2015parameterized} for the undefined notions of treewidth and nice tree decompositions. We write $G_i$ for the subgraph induced by all vertices in the bags of the subtree rooted at $i$. For each node $i \in V(T)$, we compute a table $\soli$. 
The entries of the table are indexed by four-tuples. We refer to the four-tuples as \emph{states}. A state $s$ at the node $i$ is hence a tuple $(q, \ell, \valpha, \vbeta)$, where

\begin{itemize}
    \item $0 \leq q \leq k$ is an integer that specifies the size of the desired vertex cover solution (the number of tokens) for the sub-problem on $G_i$, 
    \item $0 \leq \ell \leq \budget$ is an integer that specifies the number of sliding steps allowed in the sub-problem on $G_i$,
    \item $\valpha: X_i \to \{0,1\}$ is an indicator function for the vertices in $X_i$ (marking the vertices of $X_i$ that must belong to a solution), and
    \item $\vbeta: X_i \to [-k,k]$ is a function that specifies the number of tokens that slide through each vertex in $X_i$, where a negative value indicates that more tokens slide from~$G_i$ to the rest of the graph, while a positive value indicates the converse. 
    Let $\gamma^+ = \{u \in X_i \mid \vbeta(u) > 0\}$ and $\gamma^- = \{u \in X_i \mid \vbeta(u) < 0\}$. 
    For each vertex $u \in \gamma^+$, we can use $\vbeta(u)$ many additional tokens available at $u$. 
    These tokens are assumed to be available at the vertex $u$ which are slid from the vertices outside $G_i$. 
    Further, for each vertex $u \in \gamma^-$, we have to slide $\vbeta(u)$ many tokens to the vertex $u$ where the purpose of these tokens is that they can later slide to vertices outside $G_i$.
\end{itemize}

We denote by $\Si$ the set of all states of the node $i$. 
For a state \mbox{$s=(q, \ell, \valpha, \vbeta) \in \Si$}, $\solis$ is set to 1 if and only if the following holds: 
There exists a set $C \subseteq V(G_i)$ of size~$q$ such that

\begin{itemize}
    \item $C$ is a vertex cover of $G_i$ where $C \cap X_i = \valpha^{-1}(1)$, and
    \item $C \cup \gamma^-$ can be reached from $(S \cap V(G_i)) \cup \gamma^+$ using at most $\ell$ sliding steps, using $\vbeta(u)$ tokens from  each vertex $u \in \gamma^+$, and sliding $\vbeta(u)$ tokens to each vertex $u \in \gamma^-$.
\end{itemize}

Intuitively speaking, the function $\vbeta$ determines two quantities:

\begin{itemize}
    \item the number of available extra tokens that could reach the separator from outside $G_i$; and 
    \item the number of required tokens that need to reach the separator so that they eventually slide to vertices outside of $G_i$. 
\end{itemize} 

We compute solutions for all possible functions $\vbeta$, functions $\valpha$, solution sizes $0 \leq q \leq k$, and number of sliding steps $0 \leq \ell \leq \budget$. 

\medskip
\noindent {\bf Locally invalid states.} 
A state $s$ is said to be {\em locally invalid} if there exists a vertex $u \in X_i$ such that $N[u] \subseteq V(G_i)$ and $\vbeta(u) \not=0$. 
For any vertex $u$ with $N[u] \subseteq V(G_i)$, $\vbeta(u)$~must be $0$ since the tokens cannot slide from vertex $u$ to vertices outside of~$G_i$. 
Therefore, for a locally invalid state~$s$, we set $\solis=0$. 
Further, the existence of such a state will be no more useful for any states in the parent node. 
If a state is not locally invalid, then it is a {\em locally valid state}. 
Therefore, we consider only locally valid states in our dynamic programming computation. 

\medskip
For each node $i \in V(T)$ and a state $s \in \Si$, we show how to compute $\solis$ based on the type of the node $i$ in $T$. 

\medskip
\noindent {\bf Leaf node.}
Let $i$ be a leaf node with bag $X_i = \emptyset$. The state set $\Si$ is a singleton set with only state $s = (0, 0, \emptyset \to \{0,1\}, \emptyset \to [-k,k])$ with $\solis=1$. 

\medskip
\noindent{\bf Introduce node.} Let $i$ be an introduce node with child $j$ such that $X_i = X_j \cup \{v\}$, for some $v \notin X_j$. 
Observe that $N[v] \cap V(G_i) \subseteq X_i$. 
Let $N_v = N(v) \cap X_i$. 
If $\valpha(v) = 0$ and $N(v) \cap X_i \not\subseteq \valpha^{-1}(1)$, then we set $\solis=0$. 
That is, the edges incident on $v$ are either covered by $v$ or by the neighbors of $v$. 
Therefore, either $v$ is included in the solution or all neighbors of $v$ in $X_i$ are included in the solution. 
The vertex $v$ will be treated as the external vertex when we process the table entries of node $j$. 
We consider all possible ways of sliding tokens through $v$. 
Let $F = \{f: N_v \to [-k, k] \mid \sum_{u \in N_v}f(u) = \vbeta(v)\}$. 
A function $f \in F$ defines the number of tokens given to (or obtained from) the vertex $v$ by each neighbor $u \in N_v$. 
For each $f \in F$, we define $\vbeta_f: X_j\to[-k,k]$ such that, for each $u \in X_j$, we have 

\begin{equation*}
    \vbeta_f(u) = \begin{cases} \vbeta(u) + f(u) & \text{if } u \in N_v, \\ 
\vbeta(u) & \text{otherwise.}\end{cases}
\end{equation*}

For each $f\in F$, let $den(f) = \sum_{u \in N_v} |f(u)|$ and $s_f = (q-\valpha(v), \ell - den(f), \valpha_j, \vbeta_f)$. 
The value $den(f)$ refers to the number of sliding steps required to slide tokens through or to $v$ and $\valpha_j$ denotes the function $\valpha$ restricted to the set $X_j$. 
Finally, the table value for the state $s$ can be computed as

\begin{equation*}
    \solis = \bigvee_{f \in F} \solj[s_f].
\end{equation*}

The time required to compute the value $\solis$ is bounded by $\mathcal{O}(|F|\cdot t) = 2^{\mathcal{O}(t \log k)}$. 

\medskip
\noindent{\bf Forget node.} Let $i$ be a forget node with child $j$ such that $X_i = X_j \setminus \{v\}$, for some $v \in X_j$. 
Observe that $N[v] \subseteq V(G_i) = V(G_j)$. 
Therefore, for any locally valid state $s'=(q', \ell', \valpha', \vbeta') \in \S_j$, $\vbeta'(v) = 0$. 
Let us define $\vbeta_j: X_j \to [-k, k]$ such that,  for each $u \in X_j$, we have 
\[\vbeta_j(u) = \begin{cases} 0 & \text{if } u = v, \\ \vbeta(u) & \text{otherwise.}\end{cases}\]

Next, we try both possibilities of either including or excluding the vertex $v$ from the desired solution. 
For $c \in \{0,1\}$, let us define $\valpha_c: X_j \to \{0,1\}$ such that, for each $u \in X_j$, we have 
\[\valpha_c(u) = \begin{cases} c & \text{if } u = v, \\ \valpha(u) & \text{otherwise.}\end{cases}\]

We compute the value $\solis$ for the state $s$ as

\begin{equation*}
    \solis = \bigvee_{c \in \{0,1\}}\solj[(q, \ell, \valpha_c, \vbeta_j)].
\end{equation*}

The value $\solis$ can be computed in constant ($\mathcal{O}(1)$) time. 

\medskip
\noindent{\bf Join node.} Let $i$ be a join node with children $j$ and $h$ such that $X_i = X_j = X_h$. 
Let $F = \{f:X_i \to [-k,k]\}$. 
The size of the desired solution is $q$, where $\valpha^{-1}(1)$ vertices are contained in $X_i$. 
Therefore, we try all possible ways to get the $q-\valpha^{-1}(1)$ additional vertices from $V(G_i) \setminus X_i$. 
For each $0 \leq q' \leq q-\valpha^{-1}(1)$, let $q_j = q' + \valpha^{-1}(1)$ and $q_h = q-q'$. 

Similarly, we have a budget of $\ell$ sliding steps in total. Therefore, we try all possible ways of splitting $\ell$ into the two subtrees rooted at $j$ and $h$, respectively, where we assign~$\ell'$ to the subtree rooted at $j$ and $\ell - \ell'$ to the subtree rooted at $h$. 
Next, we define the computation of $\vbeta$ functions for $j$ and $h$. 
For each $f \in F$, we define $\vbeta_j^f: X_j \to [-k,k]$ and $\vbeta_h^f:X_h\to[-k,k]$ as follows. 
For each $u \in X_j$, $\vbeta_j^f(u) = \vbeta(u) + f(u)$, and for each $u \in X_h$, $\vbeta_h^f(u) = \vbeta(u) - f(u)$. 
A function $f \in F$ defines the number of tokens that move across the two subtrees (rooted at $j$ and $h$). A positive sign represents token movement  from the subtree rooted at $j$ to the subtree rooted at $h$, and vice versa for a negative sign. 
We compute the value $\solis$ for the state $s$ as 
\begin{equation*}
    \solis = \bigvee_{\substack{0 \leq q' \leq q-\valpha^{-1}(1), \\ 0 \leq \ell' \leq \ell, \\ f \in F}} (\solj[q_j, \ell', \valpha, \vbeta_j^f] \wedge \solh[q_h, \ell-\ell', \valpha, \vbeta_h^f]).
\end{equation*}

The time to compute the value of $\solis$ is bounded by $\mathcal{O}(|F|\cdot k \cdot \budget) = 2^{\mathcal{O}(t \log k)}$. 

\smallskip
Given a graph $G$ on $n$ vertices, a tree decomposition of $G$ has $\mathcal{O}(n)$ nodes. 
For each node~$i$ in the tree decomposition, the table $\soli$ stores values for $2^{\mathcal{O}(t \log k)}$ many states. 
A table entry for a state can be computed in time $2^{\mathcal{O}(t \log k)}\cdot n^{\mathcal{O}(1)}$.
Therefore, the tables for all the nodes in the tree decomposition can be computed in time $2^{\mathcal{O}(t \log k)}\cdot n^{\mathcal{O}(1)}$. 
Let~$r$ be the root node. 
We obtain the final answer from the table entry at the root node. 
That is, $(G, S, \budget)$ is a yes-instance of \textsc{VCD} if and only if $\sol_r[(k,\ell, \emptyset\to\{0,1\}, \emptyset\to[-k,k])] = 1$ for some $\ell \leq \budget$. 
\end{proof}

We now consider the parameter $\budget$ on (structurally) nowhere dense classes of graphs. 
The notion of nowhere dense graph classes has been introduced as a common generalization of several previously known notions of sparsity in graphs such as planar graphs, graphs with forbidden (topological) minors, graphs with (locally) bounded treewidth, or graphs with bounded maximum degree. 
It is known that for every nowhere dense class of graphs there is a function $f$ such that every property definable as a first-order formula~$\phi$ can be decided in time $f(|\phi|, \epsilon) \cdot n^{1+\epsilon}$, for any $\epsilon > 0$~\cite{DBLP:conf/stoc/GroheKS14}. 
Recently, this tractability result was extended to structurally nowhere dense classes~\cite{dreier2023first}, which constitute a dense version of nowhere dense classes. 
We use this meta-theorem to show the existence of a fixed-parameter tractable algorithm for the {\sc Vertex Cover Discovery} problem parameterized by $\budget$ and restricted to structurally nowhere dense graph classes. 

\begin{theorem}\label{thm-vc-nowhere}
The \textsc{Vertex Cover Discovery} problem is fixed-parameter tractable when parameterized by $\budget$ and restricted to structurally nowhere dense classes of graphs. 
\end{theorem}

\begin{proof}
We reduce \textsc{VCD} to the first-order model checking problem, which is known to be fixed-parameter tractable on structurally nowhere dense graph classes~\cite{dreier2023first}. 
It suffices to show that we can construct a first-order formula~$\phi$ expressing the existence of a solution to a \textsc{VCD} instance such that $|\phi| = f(\budget)$, for some function $f$. 
Given an instance $(G,S,\budget)$ of \textsc{VCD}, we mark the set~$S$ by a unary predicate, that is, a unary relation such that $S(v)$ holds whenever $v \in S$. This makes the input configuration accessible to the first-order formula. 
To construct $\phi$ we proceed as follows. Note that we need to express the existence of $q \leq \min\{k,\budget\}$ distinct vertices in $S$ that can reach $q$ other distinct vertices not in $S$ by paths of total length at most $\budget$ such that the resulting set is a vertex cover of $G$. 
To express this, we existentially quantify over at most $q \leq \min\{k,\budget\}$ distinct vertices $u_1$, $\ldots$, $u_{q}$ that belong to $S$ and we existentially quantify over $q$ distinct vertices~\mbox{$v_1$, $\ldots$, $v_{q}$} that belong to $V(G) \setminus S$.
Now, for all $\ell_1,\ldots, \ell_q$ such that $\sum_{1\leq i\leq q}\ell_i\leq \budget$ we verify (again by existential quantification) that there exists a path from $u_i$ to $v_i$ of length $\ell_i$. We make a big disjunction over all such choices for $\ell_1,\ldots, \ell_q$. 
Finally, it remains to check that the set $(S \setminus \{u_i \mid i \in [q]\}) \cup \{v_i \mid i \in [q]\}$ is a vertex cover of the graph. 
This can be accomplished by checking that each edge of the graph has at least one endpoint either in $S\setminus \{u_1,\ldots, u_q\}$ or in $\{v_1,\ldots, v_q\}$. 
We omit the formal presentation of the formula as it offers no further insights. 
\end{proof}

\section{Independent set discovery}\label{sec:isd}

In the \textsc{Independent Set (IS)} problem, we are given a graph~$G$ and an integer $k$ and the problem is to decide whether~$G$ contains an independent set of size at least $k$, where an independent set is a set of pairwise non-adjacent vertices. 

\subsection{Related work}
The \textsc{Independent Set} problem is \textsf{NP}-complete~\cite{karp1972reducibility}, and even \textsf{NP}-complete to approximate within a factor of $n^{1-\epsilon}$, for any $\epsilon>0$~\cite{zuckerman2006linear}. 
The \textsc{Independent Set} problem parameterized by solution size~$k$ is \textsf{W[1]}-complete and hence assumed to not be fixed-parameter tractable~\cite{downey1995fixed}. 

Hence, on general graphs, local search approaches cannot be expected to improve the above stated approximation factor. 
However, in practice we are often dealing with graphs belonging to special graph classes, e.g.\ planar graphs and, more generally, classes with subexponential expansion, where local search is known to lead to much better approximation algorithms, and even to polynomial-time approximation schemes (PTAS), see e.g.~\cite{har2017approximation}. 


The \textsc{Independent Set} problem is also one of the  most studied problems under the combinatorial reconfiguration framework~\cite{van2013complexity,nishimura2018introduction}. 
With respect to classical complexity, results for the \textsc{Independent Set Reconfiguration} problem and the
\textsc{Vertex Cover Reconfiguration} problem are interchangeable; we simply complement vertex sets (recall that each edge has at least one endpoint in a vertex cover, so that the remaining vertices form an independent set). More positive results (quite different than those for~\textsc{Vertex Cover Reconfiguration}) are possible if we consider the parameterized complexity of the problem~\cite{DBLP:journals/algorithmica/MouawadN0SS17}. 

\subsection{Notation and definitions}

In the \textsc{Independent Set Discovery (ISD)} problem, we are given a graph $G$, a starting configuration $S$ consisting of~$k$~tokens, and a positive integer budget $\budget$. 
The goal is to decide whether we can reach an independent set of $G$ starting from~$S$ using at most $\budget$ token slides (we do not allow two or more tokens to occupy the same vertex). Again, we denote an instance of \textsc{ISD} by $(G,S,\budget)$ and make the considered parameter explicit in the text. 

Note that, for \textsc{Independent Set Discovery}, the token jumping model and token addition/removal model boil down to the following problems. 
If $k\leq \budget$ for token jumping or $k \leq 2\budget$ for token/addition removal, then the question is simply whether there exists a solution of size $k$, as in this case we can simply move the tokens to this solution one by one. In other words, the problem boils down to the classical \textsc{Independent Set} problem. 
If $\budget \leq k$ (for token jumping) or $2\budget \leq k$ (for token addition/removal), then the question is whether there exists a solution whose symmetric difference with the initial configuration is at most~$2\budget$. 
This question has been studied in the local search version of \textsc{Independent Set}. 
We therefore focus on the token sliding model. 


\subsection{Intractability}
We first show \textsf{NP}-completeness of \textsc{Independent Set Discovery} even restricted to planar graphs of maximum degree four.

\begin{theorem}
\label{thm:ISD_NPcomplete}
The \textsc{Independent Set Discovery} problem is \textsf{NP}-complete on the class of planar graphs of maximum degree four. 
\end{theorem}

\begin{proof}
We give a reduction from \textsc{IS} on planar graphs of maximum degree three, which is known to be \textsf{NP}-complete~\cite{mohar2001}. Given an instance $(G,\kappa)$ of \textsc{IS}, where $G$ is a planar graph of maximum degree three, we construct an instance of \textsc{ISD} as follows. 
We create a new graph $H$ that initially consists of a copy of $G$. Then, for each vertex $v \in V(H)$, we create a new path on five vertices $w_v, x_v, c_v, y_v, z_v$ and we connect $v$ to $c_v$. We choose $S = \{c_v, x_v, y_v \mid v \in V(G)\}$ and we set the budget $\budget = 2n - \kappa$, where $n = |V(G)|$. Note that $k = 3|V(G)|$. This completes the construction of the instance $(H,S,\budget)$. It is easy to observe that the graph $H$ is planar and of maximum degree four. We prove that $(G,\kappa)$ is a yes-instance of \textsc{IS} if and only if $(H,S,\budget)$ is a yes-instance of \textsc{ISD}. 

First assume that $G$ has an independent set $I$ of size at least~$\kappa$. Then, in $H$ we can slide every token on $c_v$ to $v$, where $v \in I$. For all other vertices $v \notin I$ we slide every token on~$x_v$ to $w_v$ and every token on $y_v$ to $z_v$.
Observe that we need a budget of~$2$ to repair the path on every vertex $v \notin I$, while we need only a budget of $1$ to repair the paths on vertices $v \in I$.
Since~$I$ is of size at least $\kappa$, we need no more than $2n - \kappa = \budget$ slides. To see that the resulting set is an independent set of~$H$, note that for every path on a vertex $v \in I$ we have moved the token from~$c_v$ to $v$ itself. 
As $I$ is an independent set, and the only conflicting neighbor of $x_v$ resp.\ $y_v$ is $c_v$, the tokens from these paths form an independent set. 
The tokens on paths of vertices $v \notin I$ also form an independent set. As the only neighbor of~$w_v$ is $x_v$ and the token has been moved from~$x_v$ to $w_v$, hence there is no conflict. This is also true for $y_v$ and $z_v$. As the neighbors~$x_v$ and $y_v$ of $c_v$ have been freed, and there is no token on $v$ itself, i.e., these paths form an independent set.

For the reverse direction, assume that $(H,S,\budget)$ is a yes-instance of \textsc{ISD}. 
Let $I$ be the resulting independent set. We need to show that \mbox{$|I \cap V(G)|\geq \kappa$}, which then corresponds to an independent set in $G$. 
Assume towards a contradiction that $|I \cap V(G)|=\ell<\kappa$. This implies that $3n-\ell$ tokens are still on the path gadgets. Since every path gadget can contain at most $3$ independent vertices and $I$ is an independent set, at least $n-\ell$ path gadgets contain $3$ tokens. 
It takes at least $2$ slides to keep the $3$ tokens independent while not moving them out of the path. Hence, we require a budget of at least $2n-2\ell$ for these slides. Moreover, each of the $\ell$ tokens on $V(G)$ require at least one slide. In total, we require a budget of $2n-\ell>2n-\kappa$, a contradiction. 
\end{proof}

We next show that the problem is also hard from a parameterized perspective when considering the parameter $k+\budget$. 

\begin{theorem}\label{thm:ISD-whard-general}
The \textsc{Independent Set Discovery} problem is \textsf{W[1]}-hard when parameterized by $k + \budget$ even on graphs excluding $\{C_4, \ldots, C_p\}$ as induced subgraphs, for any constant~$p$.
\end{theorem}
\begin{proof}
We present a parameterized reduction from the \textsc{Multicolored Independent Set (MIS)} problem, which is known to be \textsf{W[1]}-hard on graphs excluding $\{C_4, \ldots, C_p\}$ as induced subgraphs, for any constant $p$~\cite{DBLP:journals/algorithmica/BonnetBCTW20}. Recall that in the \textsc{MIS} problem we are given a graph~$G$ and an integer $\kappa$, where $V(G)$ is partitioned into $\kappa$ cliques \mbox{$V_1$, $V_2$, $\ldots$, $V_\kappa$}, and the goal is to find a multicolored independent set of size $\kappa$, i.e., an independent set containing one vertex from each set $V_i$, for $i \in [\kappa]$. Given an instance $(G,\kappa)$ of \textsc{MIS}, we construct an instance $(H,S,\budget)$ of \textsc{ISD} as follows. First, let~$H$ be a copy of~$G$. Then, for each $i \in [\kappa]$, we add an edge on two new vertices $\{u_i, w_i\}$ and we make $u_i$ adjacent to all vertices in $V_i$. Finally, we choose $S = \{u_i \mid i \in [\kappa]\} \cup \{w_i \mid i \in [\kappa]\}$ and we set $\budget = \kappa$. Note that $k = |S| = 2\kappa$. 

Assume that $G$ has a multicolored independent set of size $\kappa$. Let $I = \{v_1, \ldots, v_\kappa\}$ denote such a set, where $v_i \in V_i$. Then we can solve the discovery instance by sliding each token on~$u_i$ to the vertex $v_i$, as needed. 
For the reverse direction, since we need to slide all the tokens on vertices $u_i$ and each set $V_i$ can contain only one token, it follows that this is only possible if $G$   has a multicolored independent set of size $\kappa$. 
\end{proof}

In what follows, we further investigate the complexity of the \textsc{Independent Set Discovery} problem when parameterized only by $\budget$ or $k$ instead of $k+\budget$ and restricted to special graph classes. It turns out that the problem remains \textsf{W[1]}-hard when parameterized by $\budget$ and restricted to the class of 2-degenerate bipartite graphs.
\begin{theorem}
The \textsc{Independent Set Discovery} problem is \textsf{W[1]}-hard when parameterized by the budget $\budget$ even on the class of 2-degenerate bipartite graphs.
\end{theorem}
\begin{proof}
We present a parameterized reduction from the \textsc{Multicolored Clique} problem, which is a well-known \textsf{W[1]}-hard problem~\cite{cygan2015parameterized}. Recall that in the \textsc{Multicolored Clique} problem we are given a graph $G$ and an integer~$\kappa$, where $V(G)$ is partitioned into $\kappa$ independent sets $V_1$, $V_2$, $\ldots$, $V_\kappa$, and the goal is to find a multicolored clique of size $\kappa$, i.e., a clique containing one vertex from each set~$V_i$, for $i \in [\kappa]$. 
Given an instance $(G,\kappa)$ of \textsc{Multicolored Clique}, we construct an instance $(H,S,\budget)$ of \textsc{ISD} as follows. We first describe the graph~$H$. For each vertex set $V_i$, we create a new set $X_i$, where each vertex $v \in V_i$ is replaced by a path on $3$ vertices which we denote by $x_v$, $y_v$, and $z_v$. 
Moreover, we add a vertex $u_i$ that is connected to all vertices~$z_v$, for $v \in V_i$. 
Finally, we add a vertex $w_i$ that is only adjacent to vertex~$u_i$. 
For $i<j \in [\kappa]$, we use $E_{i,j}$ to denote the set of edges between vertices in $V_i$ and vertices in $V_j$ (in the graph~$G$). For each~$E_{i,j}$, we create a new set of vertices, which we denote by $Y_{i,j}$, that contains one vertex $e_{uv}$ for each edge $\{u,v\} \in E_{i,j}$. 
We additionally add an edge (consisting of two new vertices) $\{w_{i,j}, u_{i,j}\}$, where~$u_{i,j}$ is also adjacent to every vertex in $Y_{i,j}$ (but not part of $Y_{i,j})$. For each vertex $e_{uv} \in Y_{i,j}$, $i<j \in [\kappa]$, we connect $e_{uv}$ to~$y_u$ via a path consisting of two new vertices and we connect $e_{uv}$ to~$y_v$ via a path consisting of two new vertices. 
Let $X = \bigcup_{i\in[\kappa]}{X_i}$ and $Y = \bigcup_{i,j\in[\kappa]}{Y_{i,j}}$. We use $W$ to denote the set of all vertices along paths from~$Y$ to~$X$ that are at distance one from some vertex in $Y$ and we use $Z$ to denote the set of all vertices along such paths that are at distance two from some vertex in $Y$. Finally, let $S = W \cup \{u_i \mid i \in [\kappa]\} \cup \{w_i \mid i \in [\kappa]\} \cup \{u_{i,j} \mid i<j \in [\kappa]\} \cup \{w_{i,j} \mid i<j \in [\kappa]\} \cup \{y_v \mid v \in V(G)\}$ and let $\budget = 3\binom{\kappa}{2} + 2\kappa$ and $k = 2\kappa + 2\binom{\kappa}{2} + n + 2m$ (see \Cref{fig-is-hard-degenerate}).

\begin{figure}
\centering
  \begin{tikzpicture}[every node/.style={inner sep = 2pt}, scale=0.6]

  \foreach \v [count=\i] in {{0, 0}, {0, -1.5}, {0, -3.5}, {0, -5.5}, {0, -7.5}} {
	\fill[red] (\v)           circle (3pt) node(b\i) {};
	\fill      (\v)++ (0,1)   circle (2pt) node(a\i) {};
	\fill      (\v)++ (1.5,0) circle (2pt) node(c\i) {};

    \draw (a\i) -- (b\i) -- (c\i);
  }

  \draw[thin] (-.333, 1.333) rectangle (1.833, -3.833) node[right] {$X_i$};
  \draw[thin] (-.333, -4.166) rectangle (1.833, -7.833) node[right] {$X_j$};

  \node at (0.75, -2) {$\vdots$};
  \node at (0.75, -6) {$\vdots$};

  \fill[red] (3, -1.5) circle (3pt) node (u1) {} node[above,black] {$u_i$};
  \draw (c1) -- (u1);
  \draw (c2) -- (u1);
  \draw (c3) -- (u1);

  \fill[red] (4, -1.5) circle (3pt) node (w1) {} node[above,black] {$w_i$};
  \draw (u1) -- (w1);

  \fill[red] (3, -6) circle (3pt) node (u2) {} node[above,black] {$u_j$};
  \draw (c4) -- (u2);
  \draw (c5) -- (u2);

  \fill[red] (4, -6) circle (3pt) node (w2) {} node[above,black] {$w_j$};
  \draw (u2) -- (w2);

  \foreach \v [count=\i] in {{-3, 0.5}, {-3, -2.5}, {-3, -5.5}} {
	\fill (\v)          circle (2pt) node(x\i) {};
	\fill (\v)++ (0,-1) circle (2pt) node(y\i) {};

    \node[fit=(x\i) (y\i),ellipse,draw,thin] (e\i) {};

    \fill[red] (\v)++ (-1,-0.5) circle (3pt) node(u1\i) {};
    \fill[red] (\v)++ (-2,-0.5) circle (3pt) node(w1\i) {};

    \draw (u1\i) -- (x\i);
    \draw (u1\i) -- (y\i);
    \draw (u1\i) -- (w1\i);
  }

  \node[below left = 0cm and 0cm of e1] {$Y_{i,j}$};
  \node[above = 0cm of u11] {$u_{i,j}$};
  \node[above = 0cm of w11] {$w_{i,j}$};

  \foreach \u/\v in {x1/b1, y1/b1, x1/b4, y1/b5} {
    \draw (\u) -- (\v) node[pos = 0.33, fill, red, circle, inner sep = 2*0.75pt] {} node[pos = 0.66, fill, circle, inner sep = 1.5*0.75pt] {};
  }  

  \end{tikzpicture}
\caption{An illustration of the \textsf{W[1]}-hardness reduction for \textsc{Independent Set Discovery} on 2-degenerate bipartite graphs.}
\label{fig-is-hard-degenerate}
\end{figure}
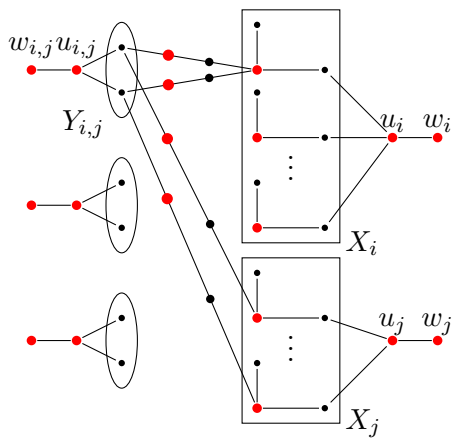

It is not hard to see that the graph $H$ is indeed bipartite by construction. To see that the graph $H$ is $2$-degenerate it suffices to note that after deleting all $y_v$ vertices, $u_i$ vertices, and $u_{i,j}$ vertices, we get a graph of maximum degree two and none of the deleted vertices are adjacent. Hence every subgraph of $H$ either has a vertex of the form $w_i$, $w_{i,j}$, $x_v$, $z_v$, or a vertex in $W \cup Y$ which all have degree at most~2, or no such vertex implying that every other vertex has degree at most~2. We claim that $(G, \kappa)$ is a yes-instance of \textsc{Multicolored Clique} if and only if $(H,S,\budget)$ is a yes-instance of \textsc{ISD}.

First assume that $(G,k)$ is a yes-instance and let $C = \{v_1, v_2, \ldots v_\kappa\}$ denote a multicolored clique in $G$, where each vertex $v_i$ belongs to $V_i$, for $i \in [\kappa]$. 
We construct a sequence of slides transforming $S$ to an independent set as follows. For each $i$, we slide the token on~$y_{v_i}$ to $x_{v_i}$ and then slide the token on $u_i$ to $z_{v_i}$. 
This requires a total of $2\kappa$ slides. Next, for each pair $i,j$ with $i < j$, we slide the token on $u_{i,j}$ to the vertex $e_{v_iv_j}$ and then slide the two tokens in $W$ to their neighbors in $Z$. 
This requires a total of $3\binom{\kappa}{2}$ slides which gives us a total of $\budget = 3\binom{\kappa}{2} + 2\kappa$ slides. Since~$C$ is a multicolored clique in $G$, it follows that the resulting configuration is indeed an independent set of $H$. 

For the reverse direction, assume that $(H,S,\budget)$ is a yes-instance of \textsc{ISD}. 
Since we have two adjacent tokens on $u_i$ and $w_i$,  for each $i \in [\kappa]$, we know that we need at least one move for each $i$. 
Moreover, since every vertex $y_v$ contains a token, we know that we need an extra slide for each $i$. 
Hence, we need a minimum of $2\kappa$ slides for the edges of the form $\{u_i, w_i\}$. Similarly, for each pair $i,j$ with $i < j$, we have two adjacent tokens on~$w_{i,j}$ and $u_{i,j}$. Moreover, all vertices in $W$ contain tokens. 
Hence, for each pair $i,j$ with $i < j$ we need at least three slides for a total of $3\binom{\kappa}{2}$ slides for the edges of the form $\{u_{i,j}, w_{i,j}\}$. Hence, there must exist $\kappa$ vertices $y_{v_i}$ and $\binom{\kappa}{2}$ vertices $e_{v_iv_j}$ adjacent to those vertices in order to successfully slide the tokens on the $u_i$ and $u_{i,j}$ vertices away from their neighbors~$w_i$ and~$w_{i,j}$ that contain tokens. 
\end{proof}

\subsection{Tractability}
Our positive results for \textsc{Independent Set Discovery} parameterized by $k$ make use of the notion of independence covering 
families introduced 
in~\cite{DBLP:journals/talg/LokshtanovPSSZ20}. 
We 
formally define such families and the various algorithms for extracting them on different graph classes.

\begin{definition}\label{def-family}
For a graph $G$ and $k \geq 1$, a family of independent sets of $G$ is called an \emph{independence covering family for $(G, k)$}, denoted by $\mathcal{F}(G, k)$,
if for every independent set~$I$ in~$G$ of size at most~$k$, there exists $J \in \mathcal{F}(G, k)$ such that~$I \subseteq J$.
\end{definition}

\begin{theorem}[\cite{DBLP:journals/talg/LokshtanovPSSZ20}]\label{thm-family-degenerate-nowhere}
Given a $d$-degenerate or a nowhere dense class of graphs $\Cc$, a graph $G \in \Cc$, and $k \geq 1$, there exists a deterministic algorithm that runs in time
$f(k) \cdot n^{\mathcal{O}(1)}$ and outputs an independence covering family for $(G, k)$ of size at most $g(k) \cdot n^{\mathcal{O}(1)}$, 
where $f(k)$ and $g(k)$ depend on $k$ and the class $\Cc$ but are independent of the size of the graph. 
\end{theorem}


We are now ready to prove the following theorem:

\begin{theorem}\label{thm-is-fpt-covering}
The \textsc{Independent Set Discovery} problem is fixed-parameter tractable when parameterized by $k$ for every class $\Cc$ of graphs that admits independence covering families of size $g(k)\cdot n^{\Oof(1)}$ computable in time $f(k)\cdot n^{\Oof(1)}$, where $f$ and $g$ are computable functions. In particular, the problem is fixed parameter tractable on $d$-degenerate and nowhere dense classes of graphs. 
\end{theorem}

\begin{proof}
Given an instance $(G,S,\budget)$ of \textsc{ISD} where $G\in \Cc$, we start by computing an independence covering family $\mathcal{F}(G, k)$ of size $g(k)\cdot n^{\Oof(1)}$ in time $f(k)\cdot n^{\Oof(1)}$, which is possible by assumption (or by \Cref{thm-family-degenerate-nowhere} for $d$-degenerate and nowhere dense classes of graphs). 
Let $\mathcal{F}(G, k) = \{J_1, J_2, \ldots\}$ denote the resulting family. 
Let $J \in \mathcal{F}(G, k)$. We construct a complete weighted bipartite graph~$H_{S,J}$ as follows. 
Let $S$ be the vertices $S$ on one side and $J$ be the vertices on the other side. We set the weight of each edge $\{u,v\}$ to be the number of edges along a shortest path from~$u$ to~$v$ in~$G$ (we set the weight to $m + \budget + 1$ whenever $u$ and $v$ belong to different components). 

It remains to show that $(G,S,\budget)$ is a yes-instance if and only if there exists a $J \in \mathcal{F}(G, k)$ such that $|J| \geq k$ and the minimum weight perfect matching in $H_{S,J}$ has weight at most $\budget$. 
Assuming the previous claim, the algorithm then follows by simply iterating over each $J$ of size at least~$k$, constructing the graph $H_{S,J}$, and then computing a minimum weight perfect matching in $H_{S,J}$. 
If we find a matching of weight at most $\budget$ then we have a yes-instance; otherwise we have a no-instance. 

Assume that $(G,S,\budget)$ is a yes-instance. Then, there exists an independent set $I$ of size $k$ that can be reached from~$S$ by at most $\budget$ token slides. 
By the definition of independence covering families, there exists $J \in \mathcal{F}(G, k)$ such that $I \subseteq J$. 
Moreover, since $I$ is reachable from~$S$ in at most $\budget$ slides, it must be the case that the weight of a perfect matching in 
$H_{S,I}$ is at most $\budget$. 
Hence, the minimum weight perfect matching in~$H_{S,J}$ has weight at most $\budget$, as needed. 

Now assume that there exists a $J \in \mathcal{F}(G, k)$ of size at least~$k$ such that the minimum weight perfect matching in~$H_{S,J}$ has weight at most $\budget$. 
Recall that, by the definition of independence covering families, $J$ is an independent set in~$G$. 
Hence, any subset $I$ of $J$ of size $k$ is an independent set of size $k$ in $G$. 
Let $I$ denote the set of vertices that are matched to some vertex in $S$ in the minimum weight perfect matching. 
As we just described, $I$ is an independent set of size $k$ in~$G$. 
Hence, it remains to show that we can reach $I$ from $S$ using at most $\budget$ slides. 
If the path a token~$t_1$ would
take to reach its destination has another token $t_2$ on it, we can have the two tokens switch destinations and continue moving $t_2$. 
One can check that the number of moves does not exceed the weight of the perfect matching. 

It is well-known that \textsc{Minimum Weight Perfect Matching} can be solved in $\mathcal{O}(n^3)$ time using either the blossom algorithm~\cite{edmonds1972theoretical} or the Hungarian algorithm~\cite{kuhn1955hungarian}.
\end{proof}

Using the same dynamic programming techniques as those used for the  \textsc{Vertex Cover Discovery} problem on graphs of bounded treewidth (\Cref{thm-vc-fpt-by-tw+k}) and reduction to first-order model checking (\Cref{thm-vc-nowhere}) on structurally nowhere dense classes, we can show the following results: 

\begin{theorem}\label{thm-is-fpt-by-tw+k}
The \textsc{Independent Set Discovery} problem can be solved in time \mbox{$2^{\mathcal{O}(t \log k)} \cdot n^{\mathcal{O}(1)}$}, where $t$ denotes the treewidth of the input graph. 
\end{theorem}

\begin{theorem}\label{thm-is-nowhere-fpt}
The \textsc{Independent Set Discovery} problem is fixed-parameter tractable when parameterized by the budget~$\budget$ and restricted to structurally nowhere dense classes of graphs. 
\end{theorem}

\section{Dominating set discovery}

In the \textsf{NP}-complete \textsc{Dominating Set (DS)} problem~\cite{GareyJ79}, we are given a graph $G$ and an integer~$k$ and the problem is to decide whether~$G$ admits a dominating set~$D$ of size at most $k$, i.e., a set~$D$ of at most $k$ vertices such that every vertex in $V(G)$ is either in~$D$ or has a neighbor in $D$. 

\subsection{Related work}
The \textsc{Dominating Set} problem can be approximated up to a factor of $\log n$ by a simple greedy algorithm~\cite{johnson1973approximation,lovasz1975ratio} and this factor asymptotically cannot be improved unless $\textsf{P} =  \textsf{NP}$~\cite{dinur2014analytical}. 
Local search by exchanging with $\lambda$-close solutions for $\lambda\in \Oof(1/poly(\epsilon))$ leads to a $(1+\epsilon)$-approximation on classes with subexponential expansion~\cite{har2017approximation}. 
The problem is \textsf{W[2]}-complete~\cite{downey1995fixed}. The dynamic variant of the problem was studied e.g.\ in~\cite{abu2015parameterized}. 

For the \textsc{Dominating Set Reconfiguration} problem, it is known that the problem under the token jumping model is \textsf{PSPACE}-complete on split graphs, bipartite graphs, graphs of bounded 
bandwidth, and planar graphs of maximum degree six~\cite{DBLP:journals/tcs/HaddadanIMNOST16}. On the positive side, the problem is linear-time solvable in  trees, interval graphs, and cographs. Bonamy et al.~\cite{DBLP:journals/dam/BonamyDO21} show that under the token sliding model \textsc{Dominating Set Reconfiguration} is \textsf{PSPACE}-complete on split, bipartite, and bounded treewidth graphs, and polynomial-time solvable on dually chordal graphs and cographs. Bousquet and Joffard~\cite{DBLP:conf/fct/BousquetJ21} show that the problem is polynomial-time solvable on circular-arc graphs and \textsf{PSPACE}-complete on circle graphs.

\subsection{Notation and definitions}
In the \textsc{Dominating Set Discovery (DSD)} problem, we are given a graph $G$, a starting configuration $S$ consisting of $k$ tokens, and a positive integer budget $\budget$. The goal is to decide whether we can reach a dominating set of $G$ of size $k$ (starting from $S$) using at most $\budget$ token slides. We denote an instance of \textsc{Dominating Set Discovery} by $(G,S,\budget)$ and make the parameter explicit in the text.

\subsection{Intractability}

We first show that the problem is \textsf{NP}-complete even on the class of planar graphs of maximum degree five. 

\begin{theorem}
\label{thm:DS_NPcomplete}
The \textsc{Dominating Set Discovery} problem is \textsf{NP}-complete on the class of planar graphs of maximum degree five. 
\end{theorem}

\begin{proof}
We give a reduction from \textsc{DS} on planar graphs of maximum degree three, which is known to be \textsf{NP}-complete~\cite{GareyJ79}. Given an instance $(G,\kappa)$ of \textsc{DS}, where $G$ is a planar graph of maximum degree three, we construct an instance of \textsc{DSD} as follows. We create a new graph $H$ which consists of a copy of $G$. Then, for each vertex $v \in V(H)$, we create a new path consisting of vertices $\{w_v, x_v, y_v, z_v\}$ and we connected $v$ to $w_v$. Then, we create an additional new vertex $u_v$ and we connect $u_v$ to both $v$ and $x_v$. 
We choose $S = \{x_v, y_v, \mid v \in V(G)\}$ and we set the budget $\budget = 2\kappa$. Note that $k = 2|V(G)|$. This completes the construction of the instance $(H,S,\budget)$ of \textsc{DSD}. It follows from the construction that the graph $H$ is planar and of maximum degree five.

Assume that $G$ has a dominating set $D$ of size at most $\kappa$. Then, in $H$ we can slide every token on $x_v$ to $v$, where $v \in D$. 
Since~$D$ is of size at most $\kappa$, we need no more than $2\kappa$ slides. To see that the resulting set is a dominating set of $H$, note that every pair of vertices $w_v$ and $u_v$ is  dominated by either $x_v$ or~$v$. Moreover, every pair of vertices~$y_v$ and~$z_v$ is dominated by~$y_v$. The vertex $x_v$ is either dominated by $x_v$ or $y_v$ (depending on whether $v \in D$ or not). Since~$D$ is a dominating set of $G$, all vertices $v$ are also dominated, as needed. 

For the reverse direction, assume that $(H,S,\budget)$ is a yes-instance of \textsc{DSD}. 
Note that moving a token on $x_v$ to either $w_v$ or $u_v$ leaves a none dominated vertex. Moreover, to dominate all vertices $\{u_v, w_v, x_v, y_v, z_v\}$ we need at least two tokens which implies that in the resulting configuration, for each $v$, we either have a token on $v$ or a token on $x_v$; the second token is either on $y_v$ (if there is a token on $v$ and no token on $x_v$) or possibly $e_v$ (if the token on $x_v$ does not slide to $v$). Since $\budget = 2\kappa$ and the distance from $x_v$ to $v$ is two, we can have at most $\kappa$ tokens slide to vertices corresponding to vertices of $G$. Such vertices must form a dominating set of $G$, as needed. 
\end{proof}

We next show that the problem is also hard from a parameterized perspective when parameterized by $k+\budget$. 
\begin{theorem}
The \textsc{Dominating Set Discovery} problem is \textsf{W[2]}-hard when parameterized by $k + \budget$ even on the class of bipartite graphs.
\end{theorem}

\begin{proof}
We present a parameterized reduction from the \textsc{DS} problem, which is known to be \textsf{W[2]}-hard on general graphs. Given an instance $(G,\kappa)$ of \textsc{DS}, we construct an instance $(H,S,\budget)$ of \textsc{DSD} as follows. First, $H$ contains two copies of the vertex set of~$G$. We denote these two sets by~$L$ and~$R$. We connect each vertex $v \in L$ to $N[v] = \{v' \in R \mid \{v,v'\} \in E(H)\} \cup \{v\}$ in~$R$, that is, we connect $v$ to each vertex in its closed neighborhood in $R$. Then, we set $k=\kappa +1$ and add $k$ new vertices $v_1$, $\ldots$, $v_k$, such that each $v_i$ is adjacent to all vertices in $L$. We further add one more vertex $u$ such that~$u$ has $\kappa + 1$ pendent neighbors $\{w_1, \ldots, w_{\kappa + 1}\}$ attached to it. We make $u$ adjacent to all vertices in $L$. We choose $S = \{u, v_1, v_2, \ldots, v_k\}$ (hence $k = \kappa + 1$) and we set $\budget = \kappa$. This completes the construction. The graph $H$ is indeed bipartite with bipartition $(L \cup \{w_i \mid i \in [\kappa + 1]\}, R \cup \{u\} \cup \{v_i \mid i \in \kappa\})$.

Assume that $G$ has a dominating set of size $\kappa$. Let $D$ denote such a set. Then we can solve the discovery instance by sliding each token on $v_i$ to some copy of a vertex of $D$ in $L$. 
Those vertices will dominate $R \cup \{v_i \mid i \in \kappa\}$ while vertex $u$ dominates $L \cup \{w_i \mid i \in [\kappa + 1]\}$. For the reverse direction, we can assume without loss of generality that, in the resulting configuration, we have one token on $u$ and no tokens on $R$. 
This follows from the fact that not having a token on $u$ requires moving $\kappa + 1 > b$ tokens to its neighbors $\{w_1, \ldots, w_{\kappa + 1}\}$, which is impossible. Moreover, every token in $R$ can only dominate itself since $L$ is already dominated by $u$. 
Hence, every token in $R$ can instead be placed using less moves on the copy of the same vertex in $L$ and potentially dominate more vertices. Putting it all together, we can assume that the resulting configuration contains a token on $u$ and at most $\kappa$ tokens in $L$ that must dominate all vertices of~$R$, as needed.
\end{proof}

For the next result, we use the standard reduction from \textsc{Vertex Cover} to \textsc{Dominating Set} to reduce \textsc{Vertex Cover Discovery} on $2$-degenerate bipartite graphs to \textsc{Dominating Set Discovery} on $2$-degenerate graphs. 

\begin{theorem}
The \textsc{Dominating Set Discovery} problem is \textsf{W[1]}-hard when parameterized by the budget $\budget$ even on the class of $2$-degenerate graphs.
\end{theorem}

\begin{proof}
Given an instance $(G,S,\budget)$ of \textsc{VCD}, we create an instance $(H,S,\budget)$ of \textsc{DSD} where $H$ is obtained from $G$ by adding a new vertex $e_{uv}$ for each edge $\{u,v\} \in E(G)$ and connecting $e_{uv}$ to both $u$ and $v$. Since we do not increase the distances between vertices of $G$ and we can assume, without loss of generality, that a dominating set of~$H$ will not contain subdivision vertices, i.e., vertices of the form $e_{uv}$, the proof follows. Note that since we start with a $2$-degenerate bipartite graph $G$ and we simply add vertices of degree two (forming triangles) it follows that $H$ is also $2$-degenerate (but not bipartite).
\end{proof}

\subsection{Tractability}

Our positive results for \textsc{Dominating Set Discovery} parameterized by $k$ make use of the notion of $k$-domination cores, a notion that to the best of our knowledge goes back to Dawar and Kreutzer~\cite{DawarK09}. 

\begin{definition}
Let $G$ be a graph and $k\geq 1$. A set $C\subseteq V(G)$ is a \emph{$k$-domination core} if every set of size at most $k$ that dominates $C$ also dominates $G$. 
\end{definition}

Bounded size domination cores do not exist for general graphs, however, they do exist for many important graph classes, see e.g.~\cite{kreutzer2018polynomial,telle2019fpt}, most generally for semi-ladder free graphs~\cite{FabianskiPST19}. 

\begin{theorem}[\cite{FabianskiPST19}]\label{thm-dom-core}
Let $\Cc$ be a class of graphs with bounded semi-ladder index. Then there exists a function $f$ such that for every graph $G\in \Cc$ and $k\in \mathbb{N}$ we can compute in polynomial time a $k$-domination core of $G$ of size $f(k)$.  
\end{theorem}

\begin{theorem}
The \textsc{Dominating Set Discovery} problem is fixed-parameter tractable when parameterized by $k$ and restricted to semi-ladder-free graphs. 
\end{theorem}

\begin{proof}
Let $(G,S,\budget)$ denote an instance of \textsc{DSD}. 
We first compute a domination core~$C$ of size~$f(k)$, which is possible due to \Cref{thm-dom-core}. 
We then compute the projection classes towards $C$. That is, we compute a family of sets/classes of vertices such that any two vertices $u,v \in V(G) \setminus C$ belong to the same set/class if and only if $N(u) \cap C = N(v) \cap C$. The number of projection classes is also bounded in terms of $k$ only; in fact a trivial upper bound is $2^{f(k)}$ but better bounds are known. 
By the definition of domination cores, we can assume, without loss of generality, that any minimal dominating set of $G$ of size at most $k$ contains at most one vertex from each projection class. 
Hence, we can now enumerate all minimal dominating sets of size at most $k$ by treating each projection class as a single vertex. Since both $|C|$ and the number of projection classes is bounded by a function of $k$, this brute-force enumeration can be accomplished in time bounded by some function of $k$. Let $D$ denote a dominating set consisting of vertices from $C$ as well as vertices representing projection classes. 
We now construct the following bipartite graph~$H$. On one side we have the vertices of $S$ and on the other side we have the vertices of $D$. We connect each $u \in S$ to a vertex $v \in D \cap C$ and assign as a weight the length of a shortest path from $u$ to $v$.  
We connect each $u \in S$ to a vertex $v \in D \setminus C$ and assign as a weight the length of a shortest path from $u$ to any vertex in the corresponding projection class; the weight being zero if~$u$ belongs to the projection class. The rest of the proof proceeds just as in the case of \textsc{ISD} or \textsc{VCD}, i.e., we look for perfect matching of weight at most $\budget$ whose existence implies that can we slide tokens along edges of shortest paths between pairs of matched vertices. 
\end{proof}

Again, using the same dynamic programming techniques as those used for the \textsc{Vertex Cover Discovery} problem on graphs of bounded treewidth (\Cref{thm-vc-fpt-by-tw+k}) and reduction to first-order model checking (\Cref{thm-vc-nowhere}) on structurally nowhere dense classes, we get the following results: 

\begin{theorem}\label{thm-ds-fpt-by-tw+k}
The \textsc{Dominating Set Discovery} problem can be solved in time \mbox{$2^{\mathcal{O}(t \log k)} \cdot n^{\mathcal{O}(1)}$}, where $t$ denotes the treewidth of the input graph. 
\end{theorem}

\begin{theorem}
The \textsc{Dominating Set Discovery} problem is fixed-parameter tractable when parameterized by $\budget$ and restricted to structurally nowhere dense classes of graphs. 
\end{theorem}

\section{Coloring discovery}\label{sec-coloring}
Until now the feasible solutions for all considered base problems are subsets of vertices satisfying some problem specific properties. In fact, our framework is much richer and can be applied to a broader class of problems. We exemplify the richness of our solution discovery framework by turning to the very fundamental \textsc{Coloring} problem, where possible solutions are no longer subsets of vertices.

A \emph{$k$-coloring} of a graph $G$ is a mapping $\phi\colon V(G) \rightarrow [k]$. A $k$-coloring is said to be \emph{proper} if whenever $\{u, v\} \in E(G)$ then $\phi(u) \neq \phi(v)$. In the \textsf{NP}-complete \textsc{Coloring} problem~\cite{GareyJ79}, we are given a graph $G$ and an integer $k$ and the goal is to decide whether~$G$ admits a proper $k$-coloring. 

\subsection{Related work}

The reconfiguration variant of the \textsc{Coloring} problem is a central problem in the combinatorial reconfiguration literature~\cite{van2013complexity,nishimura2018introduction}. 
There are many possible definitions of adjacency relations between feasible (and infeasible) colorings. The reconfiguration steps we consider (cf.\ token models) are defined next. 
%
In the \emph{color flipping model}, in each step, we can change the color of any vertex to any color in the color set $[k]=\{1, \dots, k\}$.
In the \emph{color swapping model}, a $(u,v)$-color-swap allows for the swap of colors between two  arbitrary vertices $u \in V(G)$ and $v \in V(G)$, where $u \neq v$. 
Finally, in the \emph{color sliding model}, a $(u,v)$-color-slide is a $(u,v)$-color swap along an edge of the graph. In other words, in color sliding we can only perform a $(u,v)$-color-swap if $\{u, v\} \in E(G)$.

The color flipping model is one of the most studied models in the reconfiguration variant of the \textsc{Coloring} problem. Rather surprisingly, in this model the \textsc{Coloring Reconfiguration} problem is solvable in polynomial time for $k\leq 3$~\cite{cereceda2011finding} and \textsf{PSPACE}-complete for $k\geq 4$~\cite{bonsma2009finding}. The color sliding model was studied in~\cite{bonnet2018complexity}. 
To the best of our knowledge, the color swapping model has not been considered in the reconfiguration framework. 

The discovery variant of the \textsc{Coloring} problem has already been studied (under different names) in the color flipping model (\textsc{$k$-Fix}~\cite{garnero2018fixing}) and the color swapping model (\textsc{$k$-Swap}~\cite{de2019complexity}). Garnero et al.~\cite{garnero2018fixing} show that, for color flipping, the discovery variant is \textsf{NP}-complete, even for bipartite planar graphs. 
Moreover, they show that the problem is \textsf{W[1]}-hard when parameterized by~$\budget$, even for bipartite graphs, whereas it is fixed-parameter tractable when parameterized by $k + \budget$. Interestingly, the latter is not true in the color swapping model, where the problem is \textsf{W[1]}-hard for any fixed $k\geq 3$ when parameterized by $\budget$~\cite{de2019complexity}. 
It is also shown that both problems \textsc{$k$-Fix} and \textsc{$k$-Swap} are \textsf{W[1]}-hard when parameterized by the treewidth of the input graph.
Finally, it is shown that the \textsc{$k$-Fix} problem, i.e., discovery in the color flipping model,  is polynomial-time solvable for $k\leq 2$ colors but this question was left open for the color swapping and color sliding models~\cite{garnero2018fixing}. Related to our work is also the notion of chromatic villainy~\mbox{\cite{clark2006chromatic}}. 
Given the above, we mainly consider (unless stated otherwise) \textsc{Coloring Discovery} in the color sliding model.

\subsection{Notation and definitions}

In the \textsc{Coloring Discovery (CD)} problem, we are given a (non-proper) $k$-coloring $\phi\colon V\rightarrow [k]$ of a graph $G=(V,E)$ and a budget $\budget \in \mathbb{N}$. The task is to decide whether there is a transformation of the given coloring into a proper $k$-coloring by using at most $\budget$ recoloring steps. 
We denote instances of \textsc{Coloring Discovery} by $(G,\phi, \budget)$ and make the parameter explicit in the text. We remark that we could add an additional input parameter $q$ if we want to reach a coloring with (at most) $q$ colors. This however is not relevant in the color sliding model that we focus on; color sliding preserves the number of colors in the graph.


\subsection{Intractability}


Even proving membership in \textsf{NP} is non-trivial in the color sliding model. We show the crucial property that, to slide the correct color to a vertex, we need at most $2n$ color slides along edges. Hence, we can assume $\budget\leq 2n^2$ and a recoloring sequence is a polynomial-time checkable certificate. 

\begin{proposition}\label{prop-in-np}
\textsc{Coloring Discovery} under the color sliding model is in \textsf{NP}.
\end{proposition}

\begin{proof} 
We show that we can assume that $\budget$ is polynomially bounded in $n$ and use a recoloring sequence as a certificate. 
In fact, we need at most $2n$ color slides along edges in order to fix the color of a vertex. 
To see this, consider a path from~$u$ to~$v$, where we aim to color~$u$ using~$v$'s color. Let the path be $(u,w_1,w_2,\dots,w_\ell,v)$. 
We swap the colors of $u$ and $v$ by the following sequence of slides:
\begin{enumerate}
    \item Start by sliding $(u,w_1),(w_1,w_2),(w_2,w_3),\dots,(w_\ell,v)$ sequentially. After performing these slides, $v$ has $u$'s original color and all other vertices have the original color of their successor on the path. 
    \item Afterwards, perform slides $(w_{\ell-1},w_\ell),(w_{\ell-2},w_{\ell-1}), \dots, (w_1,w_2), (u,w_1)$. After these slides $u$ gets the original color of $v$, $v$ still has $u$'s original color, and all other vertices are recolored to their original color. 
\end{enumerate}
Thus, if a recoloring sequence exists, at most~$2n^2$ color sliding are sufficient.
\end{proof}

We show that the \textsc{Coloring Discovery} problem is \textsf{NP}-complete for $k \geq 3$ in all three models, even on planar bipartite graphs. For the color flipping and color swapping models this has been shown recently~\cite{garnero2018fixing,de2019complexity}. Our proof unifies the hardness reduction for all three models.
We show \textsf{NP}-hardness by a reduction from the \textsc{List \mbox{$k$-Coloring} (LikC)} problem, which asks for a given graph $G= (V, E)$ and a list $L(u) \subseteq [k]$ of colors assigned to each vertex $u \in V(G)$, whether there is a proper coloring $\phi$ of $G$ with $\phi(u) \in L(u)$ for every vertex $u$ of $G$. The \textsc{LikC} problem is known to be \textsf{NP}-complete for $k\geq 3$ even for planar bipartite $3$-regular graphs~\cite{CHLEBIK20061960}.

%
\begin{theorem}\label{thm-col-npcomp}
\textsc{Coloring Discovery} with $k \geq 3$ colors is \textsf{NP}-complete under all three models, even on the class of planar bipartite graphs.
\end{theorem}

\begin{proof}
For $k\geq 3$, let $( G, L = (L_u \subseteq [k] \mid u \in V(G)))$ be an instance of the \textsc{LikC} problem, where the graph $G$ is a planar bipartite graph.
We construct an instance $(H, \phi, \budget)$ of \textsc{CD} with $\budget = n(n+1)$, where $H$ is a planar bipartite graph and $\phi$ is an initial (non-proper) $k$-coloring of its vertices. 
Let us define a function $f:[k] \to [k]$ such that $f(c) = c+1$ for $c<k$ and $f(k)=1$. We start with a graph $H$ which is a copy of $G$. 
Every vertex in $H$ is colored \textsf{k}.
For each vertex $u \in V(G)$, we add $k-1$ neighbors $u_1, \dots, u_{k-1}$ to the corresponding vertex $u$ in $H$. For each $i \in [k-1]$, the new vertex $u_i$ is colored with color \textsf{i}, respectively.
For each color $c \in L(u)$, we add $n$ neighbors to $u$ which we denote by $\{x_c^i \mid i \in [n]\}$ and which are assigned color $c$. Each vertex~$x_z^i$ will have a neighbor~$y_z^i$ which is assigned color $f(c)$. 
For each color $c \notin L(u)$, we add $\budget+1$ neighbors to $u$ denoted by $\{x_c^i \mid i \in [\budget+1]\}$ and assign them color $c$. This completes the construction of the graph $H$. 
The graph $H$ is a planar bipartite graph since $G$ is a planar bipartite graph. \Cref{fig:k3hardness} shows the corresponding gadgets for $k=3$ and the color set $\{\cgreen, \cblue, \cred\}$.

Now we prove the correctness of the reduction. 
First, we show that if $(G, L)$ is a yes-instance of the \textsc{LikC} problem then $(H, \phi, \budget)$ is a yes-instance of the \textsc{CD} problem. 
Let $\psi: V(G) \to [k]$ be a feasible solution for the instance $(G, L)$. 
We show how to find a proper coloring of $H$ using at most $\budget$ steps starting from the (infeasible) coloring $\phi$. 
For each vertex $u \in V(G)$, if $\psi(u) \not= \textsf{k}$ we swap the colors between $u$ and $u_{\psi(u)}$. 
This swap ensures that $H[V(G)]$ is properly colored. 
Also, for each vertex $u \in V(G)$ and $i \in [k-1]$ the vertices $u_i$ agree with the proper coloring locally, i.e., there is no monochromatic edge in the subgraph induced by these vertices. 
Further, for each $i \in [n]$, we swap the colors between $x_{\psi(u)}^i$ and $y_{\psi(u)}^i$. This produces a proper $k$-coloring with at most $1+n$ color swaps per vertex in $V(G)$.

We finish the proof by showing that a yes-instance of the \textsc{CD} problem corresponds to a yes-instance of the \textsc{LikC} problem. If it is possible to achieve a proper coloring in the constructed graph, we also obtain a proper $k$-coloring of the graph~$G$ by definition. We only need to show that we never use a color $c$ for some vertex~$u$ that is not in the list~$L(u)$. However, if $\phi(u)=c$, we must have recolored all vertices~$x_c^i$ for $i \in [\budget+1]$, which is a contradiction to the budget of $\budget$.

Note that given a yes-instance of \textsc{LikC} the budget of \mbox{$n(n+1)$} is sufficient in the color sliding, the color swapping, and the color flipping model. Conversely, we cannot recolor all vertices $x_c^i$ for $i \in [\budget+1]$ with a budget of $\budget$ in the color flipping and color swapping model.
\begin{figure}
    \centering
    \begin{tikzpicture}[scale=0.6, every node/.style={transform shape}]
    \makeatletter
    \tikzset{
        dot diameter/.store in=\dot@diameter,
        dot diameter=3pt,
        dot spacing/.store in=\dot@spacing,
        dot spacing=10pt,
        dots/.style={
            line width=\dot@diameter,
            line cap=round,
            dash pattern=on 0pt off \dot@spacing
        }
    }
    \makeatother

    \coordinate (x) at (-5,0);
    \coordinate (u) at ($(x) + (2,0)$);
    \coordinate (g) at ($(u) + (-1,0.5)$);
    \coordinate (b) at ($(u) + (-1,-0.5)$);
    \coordinate (r1) at ($(u) + (0,1.5)$);
    \coordinate (r2) at ($(u) + (1,1.5)$);
    \coordinate (g1) at ($(u) + (1.5,0.5)$);
    \coordinate (g2) at ($(u) + (1.5,-0.5)$);
    \coordinate (b1) at ($(u) + (0,-1.5)$);
    \coordinate (b2) at ($(u) + (1,-1.5)$);
    \coordinate (r1p) at ($(u) + (0,2.5)$);
    \coordinate (r2p) at ($(u) + (1,2.5)$);

    \draw[black, thick] (u) -- (g);
    \draw[black, thick] (u) -- (b);
    \draw[black, thick] (u) -- (r1);
    \draw[black, thick] (u) -- (r2);
    \draw[black, thick] (u) -- (g1);
    \draw[black, thick] (u) -- (g2);
    \draw[black, thick] (u) -- (b1);
    \draw[black, thick] (u) -- (b2);
    \draw[black, thick] (r1p) -- (r1);
    \draw[black, thick] (r2p) -- (r2);

    \draw[black, fill=red, thick] (u) circle (0.1cm);
    \draw[black, fill=green, thick] (g) circle (0.1cm) node[left] {$u_\cgreen$};
    \draw[black, fill=blue, thick] (b) circle (0.1cm) node[left] {$u_\cblue$};
    \draw[black, fill=red, thick] (r1) circle (0.1cm) node[left] {$x_\cred^1$};
    \draw[black, fill=red, thick] (r2) circle (0.1cm) node[right] {$x_\cred^n$};
    \draw[black, fill=green, thick] (g1) circle (0.1cm) node[above] {$x_\cgreen^1$};
    \draw[black, fill=green, thick] (g2) circle (0.1cm) node[below] {$~~~~~~x_\cgreen^{\budget+1}$};
    \draw[black, fill=blue, thick] (b1) circle (0.1cm) node[left] {$x_\cblue^1$};
    \draw[black, fill=blue, thick] (b2) circle (0.1cm) node[right] {$x_\cblue^{\budget+1}$};
    \draw[black, fill=green, thick] (r1p) circle (0.1cm) node[left] {$y_\cred^1$};
    \draw[black, fill=green, thick] (r2p) circle (0.1cm) node[right] {$y_\cred^n$};

    \draw [black, dot diameter=1pt, dot spacing=0.15cm, dots] ($(r1) + (.25, 0)$) -- ($(r2) + (-0.25, 0)$);
    \draw [black, dot diameter=1pt, dot spacing=0.15cm, dots] ($(g1) + (0, -.25)$) -- ($(g2) + (0, 0.25)$);
    \draw [black, dot diameter=1pt, dot spacing=0.15cm, dots] ($(b1) + (.25, 0)$) -- ($(b2) + (-0.25, 0)$);

    \node[black, thick] at ($(u) + (0.25, -3.5)$) {(a) ~ $L(u) = \{\cred\}$};
    
    \coordinate (x) at (0,0);
    \coordinate (u) at ($(x) + (2,0)$);
    \coordinate (g) at ($(u) + (-1,0.5)$);
    \coordinate (b) at ($(u) + (-1,-0.5)$);
    \coordinate (r1) at ($(u) + (0,1.5)$);
    \coordinate (r2) at ($(u) + (1,1.5)$);
    \coordinate (g1) at ($(u) + (1.5,0.5)$);
    \coordinate (g2) at ($(u) + (1.5,-0.5)$);
    \coordinate (b1) at ($(u) + (0,-1.5)$);
    \coordinate (b2) at ($(u) + (1,-1.5)$);
    \coordinate (r1p) at ($(u) + (0,2.5)$);
    \coordinate (r2p) at ($(u) + (1,2.5)$);
    \coordinate (b1p) at ($(u) + (0,-2.5)$);
    \coordinate (b2p) at ($(u) + (1,-2.5)$);

    \draw[black, thick] (u) -- (g);
    \draw[black, thick] (u) -- (b);
    \draw[black, thick] (u) -- (r1);
    \draw[black, thick] (u) -- (r2);
    \draw[black, thick] (u) -- (g1);
    \draw[black, thick] (u) -- (g2);
    \draw[black, thick] (u) -- (b1);
    \draw[black, thick] (u) -- (b2);
    \draw[black, thick] (r1p) -- (r1);
    \draw[black, thick] (r2p) -- (r2);
    \draw[black, thick] (b1p) -- (b1);
    \draw[black, thick] (b2p) -- (b2);

    \draw[black, fill=red, thick] (u) circle (0.1cm);
    \draw[black, fill=green, thick] (g) circle (0.1cm) node[left] {$u_\cgreen$};
    \draw[black, fill=blue, thick] (b) circle (0.1cm) node[left] {$u_\cblue$};
    \draw[black, fill=red, thick] (r1) circle (0.1cm) node[left] {$x_\cred^1$};
    \draw[black, fill=red, thick] (r2) circle (0.1cm) node[right] {$x_\cred^n$};
    \draw[black, fill=green, thick] (g1) circle (0.1cm) node[above] {$x_\cgreen^1$};
    \draw[black, fill=green, thick] (g2) circle (0.1cm) node[below] {$~~~~~~x_\cgreen^{\budget+1}$};
    \draw[black, fill=blue, thick] (b1) circle (0.1cm) node[left] {$x_\cblue^1$};
    \draw[black, fill=blue, thick] (b2) circle (0.1cm) node[right] {$x_\cblue^n$};
    \draw[black, fill=green, thick] (r1p) circle (0.1cm) node[left] {$y_\cred^1$};
    \draw[black, fill=green, thick] (r2p) circle (0.1cm) node[right] {$y_\cred^n$};
    \draw[black, fill=red, thick] (b1p) circle (0.1cm) node[left] {$y_\cblue^1$};
    \draw[black, fill=red, thick] (b2p) circle (0.1cm) node[right] {$y_\cblue^n$};

    \draw [black, dot diameter=1pt, dot spacing=0.15cm, dots] ($(r1) + (.25, 0)$) -- ($(r2) + (-0.25, 0)$);
    \draw [black, dot diameter=1pt, dot spacing=0.15cm, dots] ($(g1) + (0, -.25)$) -- ($(g2) + (0, 0.25)$);
    \draw [black, dot diameter=1pt, dot spacing=0.15cm, dots] ($(b1) + (.25, 0)$) -- ($(b2) + (-0.25, 0)$);

    \node[black, thick] at ($(u) + (0.25, -3.5)$) {(b) ~ $L(u) = \{\cred, \cblue\}$};

    \coordinate (x) at (5,0);
    \coordinate (u) at ($(x) + (2,0)$);
    \coordinate (g) at ($(u) + (-1,0.5)$);
    \coordinate (b) at ($(u) + (-1,-0.5)$);
    \coordinate (r1) at ($(u) + (0,1.5)$);
    \coordinate (r2) at ($(u) + (1,1.5)$);
    \coordinate (g1) at ($(u) + (1.5,0.5)$);
    \coordinate (g2) at ($(u) + (1.5,-0.5)$);
    \coordinate (b1) at ($(u) + (0,-1.5)$);
    \coordinate (b2) at ($(u) + (1,-1.5)$);
    \coordinate (r1p) at ($(u) + (0,2.5)$);
    \coordinate (r2p) at ($(u) + (1,2.5)$);
    \coordinate (g1p) at ($(u) + (2.5,0.5)$);
    \coordinate (g2p) at ($(u) + (2.5,-0.5)$);
    \coordinate (b1p) at ($(u) + (0,-2.5)$);
    \coordinate (b2p) at ($(u) + (1,-2.5)$);

    \draw[black, thick] (u) -- (g);
    \draw[black, thick] (u) -- (b);
    \draw[black, thick] (u) -- (r1);
    \draw[black, thick] (u) -- (r2);
    \draw[black, thick] (u) -- (g1);
    \draw[black, thick] (u) -- (g2);
    \draw[black, thick] (u) -- (b1);
    \draw[black, thick] (u) -- (b2);
    \draw[black, thick] (r1p) -- (r1);
    \draw[black, thick] (r2p) -- (r2);
    \draw[black, thick] (g1p) -- (g1);
    \draw[black, thick] (g2p) -- (g2);
    \draw[black, thick] (b1p) -- (b1);
    \draw[black, thick] (b2p) -- (b2);

    \draw[black, fill=red, thick] (u) circle (0.1cm);
    \draw[black, fill=green, thick] (g) circle (0.1cm) node[left] {$u_\cgreen$};
    \draw[black, fill=blue, thick] (b) circle (0.1cm) node[left] {$u_\cblue$};
    \draw[black, fill=red, thick] (r1) circle (0.1cm) node[left] {$x_\cred^1$};
    \draw[black, fill=red, thick] (r2) circle (0.1cm) node[right] {$x_\cred^n$};
    \draw[black, fill=green, thick] (g1) circle (0.1cm) node[above] {$x_\cgreen^1$};
    \draw[black, fill=green, thick] (g2) circle (0.1cm) node[below] {$~~x_\cgreen^{n}$};
    \draw[black, fill=blue, thick] (b1) circle (0.1cm) node[left] {$x_\cblue^1$};
    \draw[black, fill=blue, thick] (b2) circle (0.1cm) node[right] {$x_\cblue^n$};
    \draw[black, fill=green, thick] (r1p) circle (0.1cm) node[left] {$y_\cred^1$};
    \draw[black, fill=green, thick] (r2p) circle (0.1cm) node[right] {$y_\cred^n$};
    \draw[black, fill=red, thick] (b1p) circle (0.1cm) node[left] {$y_\cblue^1$};
    \draw[black, fill=red, thick] (b2p) circle (0.1cm) node[right] {$y_\cblue^n$};
    \draw[black, fill=red, thick] (g1p) circle (0.1cm) node[above] {$y_\cblue^1$};
    \draw[black, fill=red, thick] (g2p) circle (0.1cm) node[below] {$y_\cblue^n$};

    \draw [black, dot diameter=1pt, dot spacing=0.15cm, dots] ($(r1) + (.25, 0)$) -- ($(r2) + (-0.25, 0)$);
    \draw [black, dot diameter=1pt, dot spacing=0.15cm, dots] ($(g1) + (0, -.25)$) -- ($(g2) + (0, 0.25)$);
    \draw [black, dot diameter=1pt, dot spacing=0.15cm, dots] ($(b1) + (.25, 0)$) -- ($(b2) + (-0.25, 0)$);

    \node[black, thick] at ($(u) + (0.75, -3.5)$) {(c) ~ $L(u) = \{\cred, \cgreen, \cblue\}$};
    
    \end{tikzpicture}
    \caption{Gadgets for different list size of the vertices, and $k=3$}
    \label{fig:k3hardness}
\end{figure}
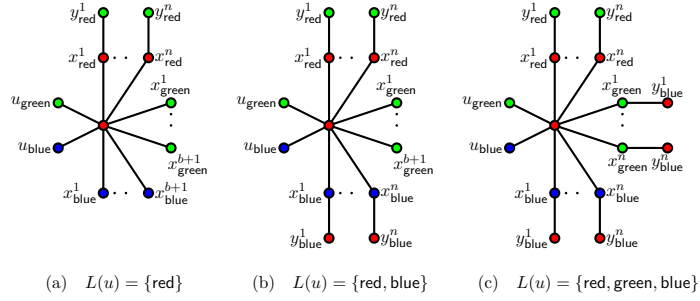
\end{proof}

%

We now consider the parameterized complexity of the \textsc{Coloring Discovery} problem with respect to the parameters $k$, $\budget$, and treewidth. The next result follows immediately from \Cref{thm-col-npcomp} and \Cref{prop-in-np}.


\begin{corollary}\label{cor-col-parahard}
\textsc{Coloring Discovery} parameterized by \mbox{$k \geq 3$} is \textsf{para-NP}-complete in all three models, even when restricted to planar bipartite graphs.
\end{corollary}

\begin{theorem}\label{thm-col-whard}
\textsc{Coloring Discovery} parameterized by $k+\budget$ is \textsf{W[1]}-hard in the color sliding model.
\end{theorem}

\begin{proof}
We reduce from the \textsc{Multicolored Independent Set} problem, known to be \textsf{W[1]}-complete. Let $(G,\kappa)$ be an instance of the problem where $V(G)$ is partitioned into independent sets $V_1,\ldots, V_{\kappa}$. We assume, without loss of generality, that each set $V_i$ has size at least $3\kappa + 1$; we can vertices appropriately otherwise. We construct an instance $(H,\phi,\budget)$ of \textsc{CD} as follows. First, let $H$ be a copy of~$G$. Then, we add new vertices $u_i$, for $1 \leq i\leq \kappa$, and new vertices $w_j$, for $1\leq j\leq \kappa + 1$. We connect $u_i$ with $w_j$, for $1 \leq i\leq \kappa$ and $1 \leq j\leq \kappa + 1$. We further add edges to turn the $w_j$ vertices, $1 \leq j\leq \kappa + 1$, into a clique of size $\kappa + 1$. For each $i \in [\kappa]$, we connect $u_i$ with every vertex in~$V_i$. We color all vertices of $V_i$ using color $i$, the vertices~$u_i$, $i \in [\kappa]$, using color~$\kappa+1$, the vertices~$w_j$, $j \in [\kappa]$, using color~$\kappa+2$, and vertex $w_{\kappa + 1}$ using color $\kappa + 1$. Hence, we have  $k=\kappa+2$ and we set the budget $\budget$ to $2\kappa$.  

Let $I = \{v_1, v_2, \ldots v_\kappa\}$ denote a multicolored independent set in $G$, where each vertex $v_i$ belongs to $V_i$, for $i \in [\kappa]$. To arrive at a proper coloring starting from $\phi$ and using $\budget = 2\kappa$ slides we proceed as follows. We first perform the slides $(v_i, u_i)$, for $i \in [\kappa]$. Next, we perform the slides $(u_i, w_i)$, for $i \in [\kappa]$. Note that every vertex $w_j$, $j \in [\kappa]$, now has a different color in $[\kappa]$, all the vertices $u_i$, $i \in [\kappa]$, have color $\kappa + 2$, vertex $w_{\kappa + 1}$ has color $\kappa + 1$, and the vertices of $I$ have color~$\kappa + 1$, which is a proper coloring of $H$.   

For the reverse direction, assume that $(H,\phi,\budget)$ is a yes-instance of \textsc{CD}. The only way to arrive at a proper coloring within $\budget = 2\kappa$  slides is to color the vertices~$\{w_j \mid j \in [\kappa + 1]\}$, using $\kappa + 1$ different colors; as these vertices form a clique. As each of those vertices has $\kappa$ neighbors colored~$\kappa+1$, we need to ``steal'' exactly $\kappa$ colors from $G$. As we must steal $\kappa$ distinct colors, each color must come from a different~$V_i$ set, $i \in [\kappa]$. Note that sliding a single color from some $V_i$ to some vertex $w_j$ requires a budget of $2$. Moreover, in the resulting proper coloring, a vertex $u_i$ cannot be colored using color $i$ since $u_i$ is adjacent to at least $3\kappa$ vertices of color $i$ (by our assumption that $|V_i| \geq 3\kappa + 1$). Moreover, no two vertices $w_j$ and $w_j'$ can receive the same color and no two vertices $u_i$ and $w_j$ can receive the same color. Putting it all together, it follows that there must exist $\kappa$ vertices, one from each $V_i$ set, that can receive color $\kappa+1$ therefore forming an independent set. This corresponds to a multicolored independent set in $G$. 
\end{proof}

\begin{theorem}\label{thm-col-tw-whard}
\textsc{Coloring Discovery} parameterized by the treewidth of the input graph is \textsf{W[1]}-hard in the color sliding model.
\end{theorem}

\begin{proof}
We follow the same strategy developed for the \textsc{$k$-Fix} and \textsc{$k$-Swap} problems in~\cite{de2019complexity}. 
In the \textsc{Precoloring Extension (PrExt)} problem, we are given a graph~$G = (V, E)$, a set $W \subseteq V(G)$ of precolored vertices, and a precoloring $c : W \rightarrow [r]$, $r \geq 2$, of the vertices in $W$. The goal is to decide whether
there is a proper $r$-coloring $c'$ of $G$ extending the coloring $c$ (i.e., $c'(v) = c(v)$ for every $v \in W$). \textsc{PrExt} is known to be \textsf{W[1]}-hard when parameterized by the treewidth of the input graph~\cite{DBLP:journals/iandc/FellowsFLRSST11}.

Let $(G, W, c)$ be an instance of \textsc{PrExt}, where we assume, without loss of generality, that $c$ is a proper $r$-coloring of $G[W]$. We construct an instance $(H,\phi,\budget)$ of \textsc{CD} as follows. First, let $H$ be a copy of $G$ and let $k = r$ denote the number of available colors. Now, we choose $\phi(v) = c(v)$ for every vertex $v \in W$. For each vertex $v \in W$, we add $(r - 1)n$ pendent vertices, where each group of $n$ vertices receives one of the colors in $[r] \setminus c(v)$. Note that, as long as we set our budget to $\budget \leq n$, the added pendent vertices guarantee that no vertex of $W$ can change colors. 
Finally, for each vertex $u \in V(G) \setminus W$, we add $r - 1$ pendent vertices, colored from $1$ to $r - 1$, and we assign vertex $u$ color $r$, i.e., all vertices $u \in V(G) \setminus W$ receive color~$r$. We set $\budget = n - |W|$, which concludes the construction. It is not hard to see that both instances are indeed equivalent and then the result follows from the well-known fact that the
addition of vertices of degree at most one does not increase the treewidth of a graph. 
\end{proof}

\subsection{Tractability}
We consider the \textsc{Coloring Discovery} problem with $k = 2$ colors. To simplify notation, we assume in this subsection the color set to be $\{\cred, \cgreen\}$. It is known that the \textsc{Coloring Discovery} problem for $k=2$ is solvable in polynomial time under the color flipping model \cite{garnero2018fixing}. We show that this is also true for the other two reconfiguration models, i.e., color swapping and color sliding.
A graph is $2$-colorable if and only if it is bipartite. We can check by a graph search in polynomial time whether a given graph $G$ is non-bipartite and declare a no-instance in this case. Hence, in what follows we restrict ourselves to bipartite graphs.

\begin{theorem}\label{thm_twoColors_swap}
\textsc{Coloring Discovery} with $k=2$ is solvable in polynomial time 
in the color swapping model. 
\end{theorem}

\begin{proof} 
For each connected component $G_i$ of the bipartite input graph $G$ 
we find a bipartition $(L_i,R_i)$ of $V(G_i)$.
%
Any proper coloring of $G_i$ colors all vertices of $L_i$ in color $\cred$ and all vertices of $R_i$ in color $\cgreen$, or vice versa. 

Both the needed budget for recoloring as well as the number of vertices colored in~$\cred$ and~$\cgreen$ in a proper coloring might be different for the two choices. 
For a given choice, both the needed number of swaps and number of colors can be calculated straightforwardly. 
As a swap always swaps colors~$\cred$ and~$\cgreen$, the overall needed budget is given by the number of vertices colored in $\cred$ that have to be recolored to $\cgreen$. Dependent on the two possibilities, we denote the needed budget by $b^i_1$ or $b^i_2$ to recolor the connected component $G_i$ properly. Additionally, for each choice we define the excess of a connected component $G_i$ by the number of vertices that are initially colored in color $\cred$ minus those that have to be colored in color $\cred$, i.e.,
        \begin{align*}
        e_1(G_i) &= |\{v \in V(G_i) \mid \phi(v) = \cred\}| - |L_i|\;, \text{ and}\\
        e_2(G_i) &= |\{v \in V(G_i) \mid \phi(v) = \cred\}| - |R_i|\;.
        \end{align*}
        
        Thus, dependent on the two possible choices we either need a budget $b^i_1$ and realize an excess~$e_1(G_i)$, or a budget of $b^i_2$ with excess $e_2(G_i)$. Overall, we can recolor the whole graph with the given budget if and only if there is a set of choices such that the overall excess is 0 and the needed budget is at most $\budget$. More formally, we can properly recolor a graph $G$ with $q$ connected components if and only if there is a vector $\vec x \in \{1,2\}^q$ with
        \begin{align*}
        &\sum_{i=1}^{q} e_{x_i}(G_i) = 0 \;, \text{ and}\\
        &\sum_{i=1}^{q} b^i_{x_i} \leq \budget\;.
        \end{align*}
        
        In the remainder of the proof, we describe a dynamic program to solve the decision problem. We define a state of the dynamic program by a triple $(x, y, z)$ of integers. This triple corresponds to the fact that there are recoloring choices for the connected components $G_1, \dots, G_x$ such that the sum of excesses is $y$ and the needed budget is $z$. We say state $(0,0,0)$ is \emph{possible}. Based on the two possible choices for each component, we reach 2 other possible states from each state. If $(x,y,z)$ is possible, $(x+1, y+e_1(G_{x+1}), z+b^{x+1}_1)$ and $(x+1, y+e_2(G_{x+1}), z+b^{x+1}_2)$ are possible. Note that there are at most $|V(G)|^3$ possible states as the number of connected components, the overall excess, and the budget can never exceed $|V(G)|$, respectively.

        By graph search for some state $(q,0,z)$, where $q$ is the number of connected components and $z\leq\budget$, we solve the problem.
\end{proof}

\begin{theorem}\label{thm_twoColors}
\textsc{Coloring Discovery} with $k = 2$ is solvable in polynomial time under  the color sliding model. 
\end{theorem}

\begin{proof}
In the color sliding model we can treat connected components independently. Thus, we assume that $G$ is connected. 
First, we compute a bipartition $(L,R)$ of $G$. There are, in principle, two possible 
proper $2$-colorings of $G$, either a vertex is colored in $\cred$ if and only if it is in $L$ or it is colored in $\cgreen$ if and only if it is in $R$. We treat both cases separately and choose the one with smallest cost. Summing over all connected components yields the result. 
We assume now that all vertices in~$L$ should be colored using color $\cred$. The other case is analogous.

We will solve the \textsc{CD} problem by computing a minimum cost maximum flow in an adjusted network. 
First, note that the number of vertices that need to be recolored in $L$ must be the same as the number of vertices that need to be recolored in $R$. 
Otherwise, the \textsc{CD} instance is not solvable. Let this number be denoted by~$q$. Given the graph $G=(V,E)$ with an initial coloring, we design the \textsc{Minimum Cost Maximum Flow} instance as follows. 
We copy the graph $G$, add a source $s$ and connect it by an arc with cost $0$ and capacity $1$ to all $v\in L$ that are colored in color~$\cgreen$, i.e., that have to be recolored. 
We add a sink $t$ with incoming arcs from each $v\in R$, that is wrongly colored, i.e., colored with $\cred$. 
These arcs have capacity $1$ and cost $0$. 
For all original edges $\{u,v\}$, we add an arc $(u,v)$ and $(v,u)$ both with cost $1$ and capacity $q$. 
We then compute a maximum flow of minimum cost from $s$ to $t$. We will show that there is a flow with flow value~$q$ and cost at most $\budget$ if and only if the \textsc{Coloring Discovery} instance is a yes-instance. 
The reduction also provides a constructive proof, i.e., we can compute a sequence of swaps yielding the recoloring. 
Note that computing a maximum flow of minimum cost can be accomplished in polynomial time~\cite{DBLP:journals/combinatorica/Tardos85}.

First, we show that if there is a flow $f$ with $|f| = q$ and cost $\budget$, we can recolor the graph and obtain a proper $2$-coloring using $\budget$ swaps. For a given flow $f$, we calculate a shortest path from $s$ to $t$ using only edges with positive flow. Let $(s, v_1, v_2, v_3, \dots, v_\ell, t)$ denote such a shortest path. As it is a shortest path it has the property that initially $v_1, v_2, v_4, v_6, \dots, v_{\ell - 2}$ are colored in color $\cgreen$, and $v_3, v_5, \dots, v_{\ell-1}, v_\ell$ are colored in color~$\cred$. Otherwise if any $v_i$, $i \geq 3$ being odd, is wrongly colored, then we have an edge from $s$ to $v_i$. Similarly, for $i \leq \ell - 2$ and $i$ being even, if~$v_i$ is wrongly colored, then we have an edge from $v_i$ to $t$, also a contradiction. 
Hence, we can swap the colors as follows. First, swap colors along the edges $(v_2, v_3), \dots, (v_{\ell-2}, v_{\ell-1})$. After these $(\ell-2)/2$ swaps, nodes $v_1, v_3, v_5, \dots, v_{\ell-1}$ are colored in $\cred$ and $v_2, v_4, \dots, v_{\ell}$ in color $\cgreen$. Now, perform color swaps along the edges $(v_1,v_2), (v_3,v_4), \dots, (v_{\ell -1}, v_\ell)$. 
After these $\ell/2$ swaps, the color of vertices $v_i$ is $\cred$ if and only if $i$ is odd. Thus, the total number of performed swaps is $(\ell-2)/2 + \ell/2 = \ell-1$, i.e., the length of the path. 
Furthermore, the swaps exactly change the colors of vertices~$v_1$ and~$v_\ell$ and all other vertices preserve their original color. Thus, all vertices on the path are colored correctly after the swaps and~$q$ is reduced by $1$. We reduce the flow by 1 along the path and continue with a next shortest path. Overall, we obtain a proper coloring after performing exactly $\budget$ swaps.

Second, we show that if we can recolor the graph by $\budget$ slides, there is a flow $f$ with $|f| = q$ and cost $\budget$ in the constructed network. Suppose we are given a slide sequence of $\budget$ slides that recolors the graph into a proper $2$-coloring. Given this sequence, we create the flow step by step by raising flow along paths from $s$ to $t$. We start with the $0$-flow. Let $e_1, e_2, \dots, e_\budget$ denote the edges of the slide  sequence. Note that some edges might appear multiple times in the sequence. Before performing the slides, we assume vertices are equipped with ids and we attach the id of the vertex to it's color. 
When sliding colors, we swap the ids together with the colors. That is, consider some vertex $v_i$ with some initial color, say $\cgreen$ in $L$. After a slide $(v_i,v_j)$ the vertex $v_i$ is colored in color~$\cred$. The id of this new color of $v_i$ points to $v_j$ that was initially colored in this color. By tracking the id during the slide procedure, we obtain a path from each vertex to the origin of the new color. For a vertex $v_i$ that received a color from $v_j$ in the sliding procedure we denote the path by $p(v_i,v_j)$. 
If $v_i \in L$ either $v_j$ was initially a wrongly colored node in $R$, or a correctly colored node in $L$. If $v_j \in R$, we stop and raise the flow along the path $(s, v_i),p(v_i,v_j), (v_j, t)$. Otherwise, $v_j \in L$, i.e., is now colored with color $\cred$. Tracking the id of the new color of $v_j$ analogously to the description above yields a path to some node $v_\ell$. We continue this procedure until we reach a node $v_{\ell'}$ in $R$. Note that the ids are unique, i.e., this procedure cannot cycle. 
By tracking the whole concatenated path $(s,v_i),p(v_i,v_j), p(v_j,v_\ell), \dots, (v_\ell', t)$ we obtain a path from $s$ to $t$ along which we raise the flow by $1$. 
We continue with the next node in $L$ that was initially colored with color $\cgreen$. Overall, we obtain a flow with flow value $q$ and cost $\budget$ as each edge with positive cost used by some flow unit corresponds to a slide in the slide procedure. 
\end{proof}

Using similar techniques to those used in \Cref{thm-vc-nowhere}, we can show the following result. Note that, unlike \Cref{thm-vc-nowhere}, here we need $k$ as part of the parameter so that our first-order formula can access all $k$ colors in an instance. We add the colors as unary predicates so that they are accessible to first-order logic. Again we quantify the $q\leq \max\{k,\budget\}$ vertices that should change their color and hardcode with a formula of length $f(\budget)$ how to arrive at the new coloring. 

\begin{theorem}
The \textsc{Coloring Discovery} problem (in all
three models) is fixed-parameter tractable when parameterized by $k + \budget$ and restricted to structurally nowhere dense classes of graphs. 
\end{theorem}

\section{Conclusion and future work}\label{sec:conclusion}
We have proved many positive and hardness results concerning solution discovery variants of fundamental graph problems. 
While some problems resulting from our framework have already been considered in the literature (albeit under different names), we believe that viewing such problems from a unified perspective can lead to more global insights and, hopefully, more unifying results. 

In future work, we plan to combine our framework with other established frameworks, e.g., dynamic problems, to categorize more of these dynamic real-world applications of decision-making problems. 
For instance, one can augment the solution discovery framework by allowing two input graphs~$G$ and $G'$ for each instance, where~$G'$ is obtained from $G$ by a small number of edge/vertex addition/deletions. 
Moreover, we require the initial state $S$ to satisfy certain properties in $G$. Now the question becomes whether we can transform $S$ to $S'$ using at most $\budget$ reconfiguration steps such that~$S'$ satisfies some/similar properties in $G'$. 
This generalization increases the degrees of freedom in which we can analyze the problems and pushes towards multivariate analyses, where the changes in the graph can now also be part of the parameter. 

Another possible avenue to explore is to relax the constraints that we impose on token positions. For instance, for graph vertex-subset problem, we can allow multiple tokens to occupy the same vertex. This, in turn, allows us to decouple the size of the starting configuration from the size of the desired target feasible solution. As an example, the \textsc{Independent Set Discovery} problem could be reformulated as follows. We are given a graph $G$, a starting configuration $S$, an integer $k \leq |S|$, and a budget $b$. The goal is now to decide whether we can reach an independent set of size at least $k$ starting from $S$ and using at most $b$ steps. Similarly, for the \textsc{Coloring Discovery} problem in the flipping model, one can add an additional input $q < k$ and ask whether we can reach a proper $q$-coloring starting from a (non-proper) $k$-coloring.

\bibliographystyle{plain}
\bibliography{ref}

\end{document}